\documentclass{article}
\usepackage{geometry}
\usepackage[utf8]{inputenc}
\usepackage{cite}
\usepackage{url}
\usepackage{amssymb,amsmath,amsthm}
\usepackage{bm}
\usepackage{graphicx}
\usepackage{enumerate}
\usepackage{microtype}
\usepackage{color}
\usepackage{hyperref}
\usepackage{multirow}
\usepackage{array}
\newcolumntype{C}[1]{>{\centering\let\newline\\\arraybackslash\hspace{0pt}}m{#1}}

\usepackage[ruled,linesnumbered,vlined]{algorithm2e}

\newtheorem{defn}{\noindent $\mathbf{Definition}$}[section]

\newtheorem{theorem}[defn]{$\mathbf{Theorem}$}

\title{Spherical Density-Equalizing Map for Genus-0 Closed Surfaces}

\author{Zhiyuan Lyu\thanks{Department of Mathematics, The Chinese University of Hong Kong
  ({zylyu@math.cuhk.edu.hk}).}
\and Lok Ming Lui\thanks{Department of Mathematics, The Chinese University of Hong Kong
  ({lmlui@math.cuhk.edu.hk}).}
\and Gary P. T. Choi\thanks{Department of Mathematics, The Chinese University of Hong Kong
  ({ptchoi@cuhk.edu.hk}).}}

\date{}

\begin{document}

\maketitle
\begin{abstract}
Density-equalizing maps are a class of mapping methods in which the shape deformation is driven by prescribed density information. In recent years, they have been widely used for data visualization on planar domains and planar parameterization of open surfaces. However, the theory and computation of density-equalizing maps for closed surfaces are much less explored. In this work, we develop a novel method for computing spherical density-equalizing maps for genus-0 closed surfaces. Specifically, we first compute a conformal parameterization of the given genus-0 closed surface onto the unit sphere. Then, we perform density equalization on the spherical domain based on the given density information to achieve a spherical density-equalizing map. The bijectivity of the mapping is guaranteed using quasi-conformal theory. We further propose a method for incorporating the harmonic energy and landmark constraints into our formulation to achieve landmark-aligned spherical density-equalizing maps balancing different distortion measures. Using the proposed methods, a large variety of spherical parameterizations can be achieved. Applications to surface registration, remeshing, and data visualization are presented to demonstrate the effectiveness of our methods. 
\end{abstract}

\section{Introduction}

Surface parameterization is the process of mapping a complicated surface onto a simple standardized domain. For genus-0 closed surfaces, it is common to consider parameterizing them onto the unit sphere. There has been a vast number of works on the computation of spherical conformal parameterizations~\cite{gu2004genus,chen2013ricci,crane2013robust,choi2015flash,choi2016spherical,yueh2017efficient,choi2020parallelizable,liao2022convergence}, which preserve angles under the mappings. There have also been some works on the computation of spherical area-preserving parameterization~\cite{angenent2000area,zou2011authalic,su2013area,pumarola20193dpeople,cui2019spherical,yueh2019novel}. Some other methods have considered the problem of achieving a balance between various distortion measures for spherical parameterization~\cite{lui2007landmark,nadeem2016spherical,choi2016fast,wang2016bijective,wang2018novel,lyu2023two}.

Density-equalizing map~\cite{gastner2004diffusion} is a widely used approach for producing cartograms, for which a planar map is deformed based on the physical principle of density diffusion. Under a density-equalizing map, regions with a larger prescribed quantity (called the \emph{population}) expand and those with a smaller population shrink, so that the area ratio of different regions in the resulting deformed map reflects the ratio of the input population. In recent years, this method has been extensively applied for the visualization of sociological and biological data~\cite{tobler2004thirty,wake2008we,glynn2010breast,pratt2012implications,matthews2014national,nusrat2016state,dodd2016global}. Some recent efforts have been made on improving the computation of density-equalizing maps on 2D or 3D grids~\cite{gastner2018fast,li2018diffusion,li2019visualization,choi2021volumetric}. More recently, several surface parameterization methods for open surfaces have been proposed using the idea of density-equalizing maps~\cite{choi2018density,choi2020area,lyu2023bijective}. However, since genus-0 closed surfaces are topologically equivalent to the sphere but not the plane, the existing planar mapping formulations are not suitable for handling such surfaces. While it is possible to use different projection methods to map planar results onto the sphere, such projections will unavoidably introduce distortions in area (e.g., the stereographic projection), angle (e.g., the Lambert equal-area projection), or both (e.g., the Mercator projection), thereby leading to inaccuracies in the final spherical density-equalizing maps. It is more natural to skip the planar parameterization and directly compute density-equalizing maps on the sphere. Also, one may want to combine density-equalization and other surface mapping criteria to achieve different desired spherical mapping effects.

\begin{figure}[t]
    \centering
    \includegraphics[width=\textwidth]{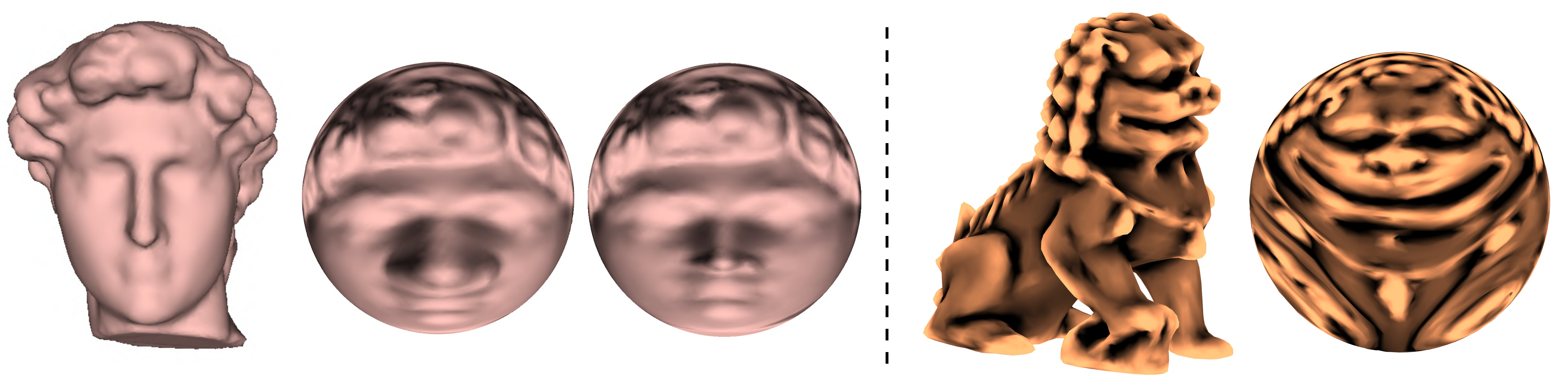}
    \caption{\textbf{Examples of the spherical density-equalizing maps obtained by our SDEM method.} (Left) The David model and two spherical density-equalizing maps obtained by our SDEM method, with the nose region enlarged or shrunk based on the prescribed population. (Right) The Chinese Lion model and the spherical area-preserving parameterization obtained by our SDEM method.}
    \label{fig:illustration}
\end{figure}

In this work, we first develop a novel method for computing bijective spherical density-equalizing maps (abbreviated as SDEM) for genus-0 closed surfaces (see Fig.~\ref{fig:illustration} for examples), in which the bijectivity of the mappings is ensured using quasi-conformal theory. Then, we propose an energy minimization model combining density-equalization, conformality, and landmark-matching conditions for producing landmark-aligned spherical density-equalizing maps (abbreviated as LSDEM) that balance between different distortion measures. We apply our proposed methods for different mapping problems and practical applications to demonstrate their effectiveness.

The rest of this paper is organized as follows. In Section~\ref{sec:background}, we review the mathematical background of this work. In Section~\ref{sec:main}, we describe our two proposed methods for spherical density-equalizing maps and landmark-aligned spherical density-equalizing maps in detail. In Section~\ref{sect:experiment}, we present numerical experiments to assess the performance of our proposed methods. In Section~\ref{sec:applications}, we introduce the applications of our methods to surface registration, surface remeshing, and spherical data visualization. We conclude our work and discuss possible future directions in Section~\ref{sect:discussion}.

\section{Mathematical background}\label{sec:background}
In this section, we describe some basic theories of density-equalizing maps and quasi-conformal geometry.

\begin{figure}[t]
   \centering
   \includegraphics[width=\textwidth]{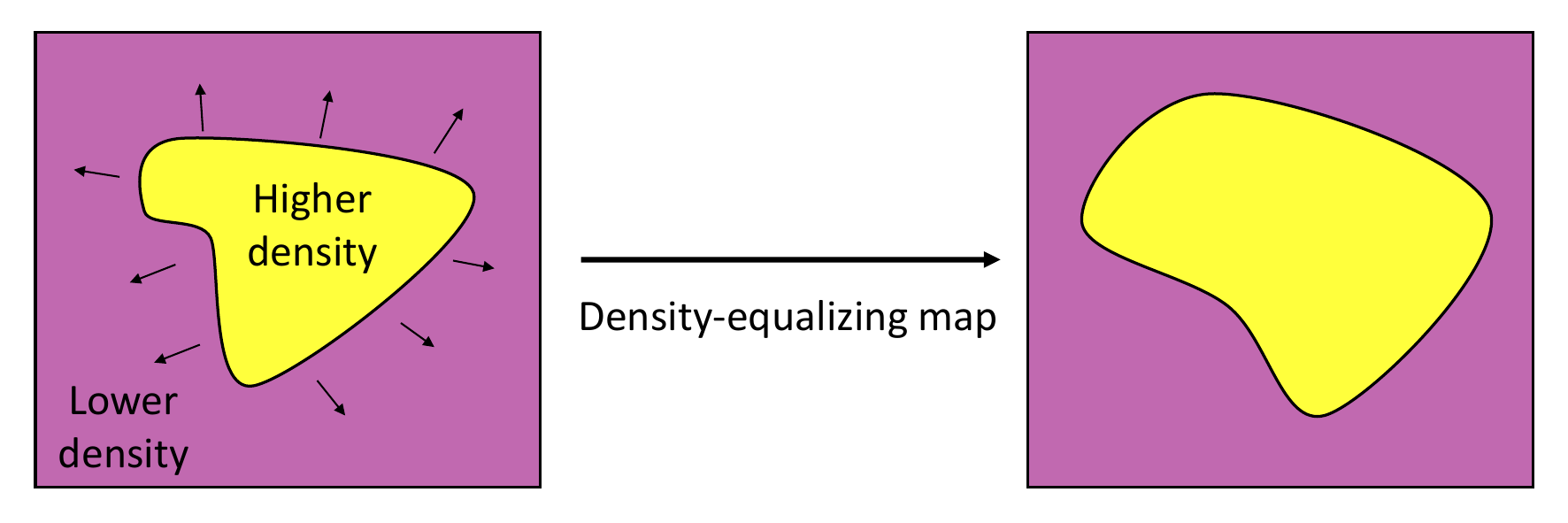}
    \caption{\textbf{An illustration of density-equalizing maps.} During the diffusion process, the regions with high density will be enlarged and the regions with low density will be shrunk. }
    \label{fig:DEM_illustration}
\end{figure}

\subsection{Density-equalizing maps}
Gastner and Newman~\cite{gastner2004diffusion} proposed a method for computing the \emph{density-equalizing maps} based on the diffusion process. Given a 2D planar domain, a positive density function $\rho$ is first defined on every unit area of it. The mapping method aims to deform the domain according to $\rho$ to obtain a density-equalizing domain. Moreover, during the deformation process, the regions with high density will be enlarged and the regions with low density will be shrunk. The advection equation is given by
\begin{equation}\label{advection}
    \frac{\partial \rho}{\partial t} = - \nabla \cdot \mathbf{j},
\end{equation}
where $\mathbf{j} = - \nabla \rho$ is the flux by Fick's law of diffusion. Hence, the diffusion equation can be obtained by
\begin{equation}\label{diffusion-eq}
    \frac{\partial \rho}{\partial t} = \Delta \rho.
\end{equation}
The expression of the velocity field in terms of density is 
\begin{equation}
    \mathbf{v} = -\frac{\nabla \rho}{\rho}.
\end{equation}
Combining the above formulas, the cumulative displacement $\mathbf{r}(t)$ of any point on the map at time $t$ can be calculated by integrating the velocity field:
\begin{equation}\label{displacement}
    \mathbf{r}(t) = \mathbf{r}(0) + \int^{t}_{0} \mathbf{v}(\mathbf{r},\tau) \mathrm{d\tau}.
\end{equation}
Note that the diffusion flow is always directed from high-density regions to low-density regions. In the limit $t \rightarrow \infty$, the density $\rho$ induced by the above displacement $\mathbf{r}(t)$ will be fully equalized per unit area. An illustration is given in Fig.~\ref{fig:DEM_illustration}.

The density-equalizing mapping method has been widely applied to cartogram creation and sociological data visualization. In recent years, Choi and Rycroft~\cite{choi2018density,choi2020area} proposed methods for computing the density-equalizing maps for general simply connected open surfaces in $\mathbb{R}^3$. More recently, Lyu, Choi, and Lui~\cite{lyu2023bijective} developed a method for computing density-equalizing quasiconformal (DEQ) maps to achieve a balance between density-equalization and quasiconformality for simply-connected and multiply-connected open surfaces.

\begin{figure}[t]
   \centering
   \includegraphics[width=\textwidth]{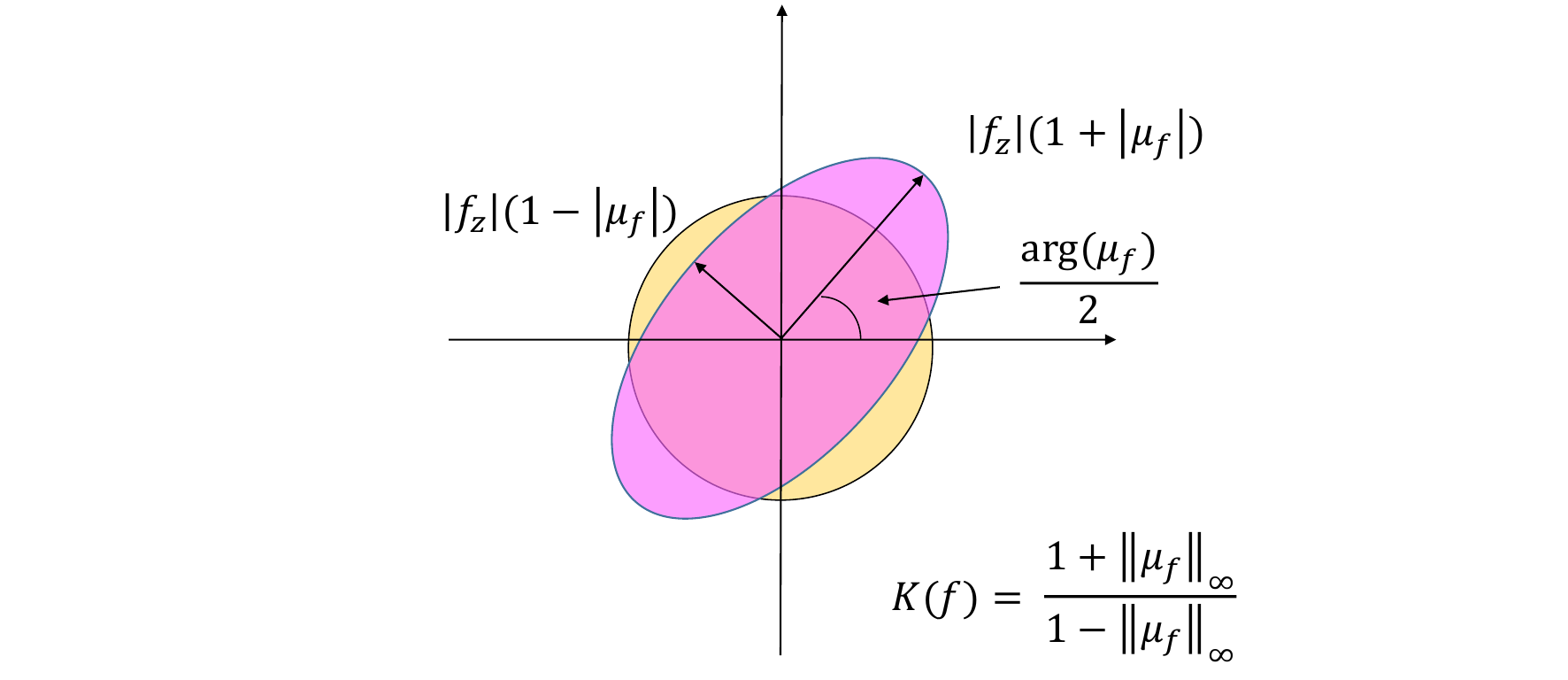}
    \caption{\textbf{An illustration of quasi-conformal maps.} Under a quasi-conformal map $f$, an infinitesimal circle is mapped to an infinitesimal ellipse with bounded eccentricity. The maximal magnification, maximal shrinkage, and maximal dilation are all related to the Beltrami coefficient $\mu$.}
    \label{fig:quasiconformal_map}
\end{figure}

\subsection{Quasi-conformal theory}
It is well-known that conformal maps preserve angles. \textit{Quasi-conformal maps}, as a generalization of conformal maps, relax the condition by allowing angle distortion within certain bounds. Mathematically, a mapping $f:\overline{\mathbb C} \rightarrow \overline{\mathbb C}$ is called a quasi-conformal map if it satisfies the Beltrami equation
\begin{equation}\label{eqt:Beltrami_eq}
    \frac{\partial f}{\partial \Bar{z}} = \mu(z) \frac{\partial f}{\partial z}
\end{equation}
for some complex-valued function $\mu$ with $\|\mu \|_{\infty}<1$. $\mu$ is called the \textit{Beltrami coefficient} of $f$, which measures the conformality distortion of the mapping $f$. In particular, if $\mu = 0$, then Eq.~\eqref{eqt:Beltrami_eq} becomes the Cauchy--Riemann equation, and hence the mapping $f$ is conformal. Intuitively, around a point $z_0 \in \mathbb C$, the first order approximation of $f$ can be expressed as:
\begin{equation}\label{eqt:first_order_approximation}
    f(z) \approx f(z_0) + f_{z}(z_0)(z-z_0) + f_{\Bar{z}}(z_0)\overline{(z-z_0)} = f(z_0) + f_{z}(z_0)(z-z_0 + \mu(z_0)\overline{(z-z_0)}).
\end{equation}
The above formula suggests that $f$ maps an infinitesimal circle centered at $z_0$ to an infinitesimal ellipse centered at $f(z_0)$. Additionally, we can determine the angles at which the magnification and shrinkage are maximized, as well as quantify the degree of magnification and shrinkage at those angles (see Fig.~\ref{fig:quasiconformal_map}). More specifically, the angle of maximal magnification is $\operatorname{arg}(\mu(z_0))/2$ with the magnifying factor $|f_{z}(z_0)|(1+|\mu(z_0)|)$, and the angle of maximal shrinkage is $(\operatorname{arg}(\mu(z_0))+\pi)/2$ with the shrinking factor $|f_{z}(z_0)(|1-|\mu(z_0)|)$. The maximal dilation of $f$ is $K(f) = \frac{1+\|\mu \|_{\infty}}{1-\|\mu\|_{\infty}}$. Thus, the Beltrami coefficient encodes important geometric information about the quasi-conformal map $f$.

Meanwhile, the Beltrami coefficient is closely related to the bijectivity of the quasi-conformal map. More specifically, we have the following result~\cite{lehto1973quasiconformal,ahlfors2006lectures}:
\begin{theorem}\label{qc_bijective}
    If $f$ is a $C^1$ mapping satisfying $\|\mu_f \|_{\infty}<1$, then $f$ is bijective. 
 \end{theorem}
 \begin{proof}
     For a given quasi-conformal map $f = u + i v$, the corresponding Jacobian matrix, denoted as $J_f$, is expressed as follows:
    \begin{equation}\label{Jacobian}
        \begin{aligned}
        J_f & =  u_{x}v_{y} - v_{x}u_{y} \\
            & = \frac{1}{4}\left( \left(u_x + v_y \right)^2 + \left(v_x - u_y \right)^2 - \left(u_x - v_y \right)^2 - \left(v_x + u_y \right)^2 \right) \\
            & = \frac{1}{4} \left( |\left(u_x + i v_x \right) - i \left(u_y + i v_y \right) |^2 - |\left(u_x + i v_x \right) + i \left(u_y + i v_y \right) |^2 \right) \\
            & =  |f_z|^2 - |f_{\Bar{z}}|^2 \\
            & =  |f_z|^2 (1 - |\mu|^2).
        \end{aligned}
    \end{equation}
    Thus, if $\|\mu \|_{\infty}<1$ and $|\frac{\partial f}{\partial z}| \neq 0$, then $J_f > 0$ everywhere. Therefore, $f$ is bijective. \hfill$\blacksquare$ 
 \end{proof}

Moreover, the Beltrami coefficient of a composition of two quasi-conformal maps can be expressed in the following way. Let $f,g: \overline{\mathbb{C}} \rightarrow \overline{\mathbb{C}}$ be two quasi-conformal maps with Beltrami coefficient $\mu_f$ and $\mu_g$, respectively. The Beltrami coefficient of the composition map $g\circ f$ is given by
\begin{equation}
    \mu_{g \circ f} = \dfrac{\mu_f + (\mu_{g}\circ f) \tau }{1 + \overline{\mu_f}(\mu_g \circ f) \tau},
\end{equation}
where $\tau = \overline{f_z} / f_z$. In particular, if $g$ is a conformal map, then $\mu_{g \circ f} = \mu_f$. In other words, composing a conformal map with a given quasi-conformal map will not change its Beltrami coefficient. This observation holds significant importance in our proposed algorithms as discussed in the following section.

\section{Proposed methods}\label{sec:main}
This section presents our main proposed methods in this work. Below, we first develop a novel algorithm for computing bijective spherical density-equalizing maps (abbreviated as SDEM) of genus-0 closed surfaces onto a unit sphere. Then, we propose another method for computing bijective landmark-aligned spherical density-equalizing maps (abbreviated as LSDEM) of genus-0 closed surfaces.

\subsection{Bijective spherical density-equalizing map (SDEM)}\label{sec:SDEM}
Let $\mathcal{M}$ be a genus-0 closed surface represented as a triangular mesh $(\mathcal{V}, \mathcal{E}, \mathcal{F})$, where $\mathcal{V}$ is the set of all vertices, $\mathcal{E}$ is the set of all edges, and $\mathcal{F}$ is the set of all triangular faces. Our goal is to compute a bijective spherical density-equalizing map $f:\mathcal{M} \to \mathbb{S}^2$ based on a prescribed population encoding the desired area changes.

\subsubsection{Initial spherical parameterization} \label{sect:sdem_initial}

First, we apply the FLASH method proposed by Choi, Lam, and Lui~\cite{choi2015flash} to compute an initial spherical conformal parameterization. FLASH is a north pole–south pole iterative scheme specifically designed to reduce the conformality distortion of the spherical conformal map using quasi-conformal theory. The process begins by selecting a triangle element from $\mathcal{F}$ as the north pole. Next, the method solves the Laplace equation on a triangular domain on the complex plane $\mathbb{C}$ to obtain a conformal map of the input surface with the selected triangle punctured. By employing the inverse stereographic projection, the planar region can then be mapped onto the sphere $\mathbb{S}^2$, and then the punctured element can be added back to the surface to get a spherical parameterization. However, the spherical mapping result often exhibits significant angle distortions near the north pole. To address this issue, the method projects the sphere onto the plane using the south-pole stereographic projection and composes the map with another quasi-conformal map to correct the conformality distortion. Finally, the updated map is projected back onto the sphere to yield the desired spherical conformal parameterization. Readers are referred to~\cite{choi2015flash} for more details.

One common issue of global conformal parameterizations is that the area may be largely distorted. In order to mitigate this problem, an additional step of composing the spherical conformal parameterization with an M\"obius transformation can be introduced, as discussed in~\cite{choi2020parallelizable}. Specifically, one can start by projecting the aforementioned spherical conformal parameterization result onto $\overline{\mathbb C}$ using the stereographic projection, and then apply an M\"obius transformation followed by the inverse stereographic projection back onto the sphere. By finding an optimal M\"obius transformation, the area distortion of the spherical parameterization can be reduced (see~\cite{choi2020parallelizable} for details). Moreover, as both the stereographic projection and M\"obius transformations are conformal, the resulting updated spherical parameterization is also conformal.

We denote the spherical conformal parameterization obtained using the above-mentioned procedures as $f_0: \mathcal{M} \rightarrow \mathbb{S}^2$.

\subsubsection{Density-equalizing map on the sphere}\label{sec:DEM_sphere}

Let $\mathbf{x}$ be a vertex on the sphere $\mathbb{S}^2$ with a positive quantity (the population) prescribed at it, giving a density $\rho$ on $\mathbb{S}^2$ defined as population per unit area. Below, we describe our proposed method for achieving density-equalizing maps on the sphere. Note that the density diffusion is a time-dependent process, and so here we denote the initial position of $\mathbf{x}$ as $\mathbf{x}(0)$. At time $t$, the updated position of the vertex is denoted as $\mathbf{x}(t)$.

\begin{figure}[t]
    \centering
    \includegraphics[width=\textwidth]{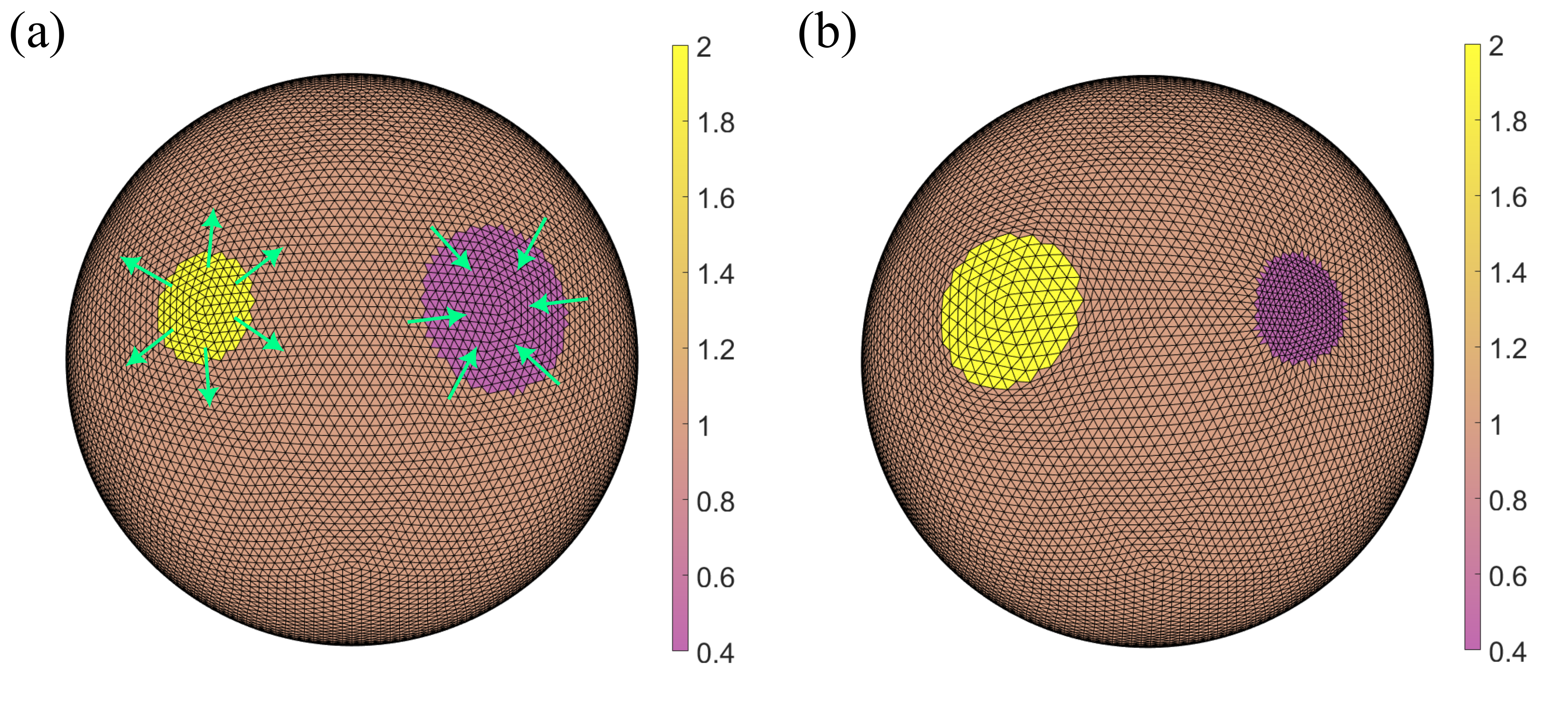}
    \caption{\textbf{Density-equalizing map on a sphere.} (a) A spherical surface color-coded with a prescribed density $\rho$. The diffusion of $\rho$ induces a density flux, thereby giving a velocity field on the sphere as indicated by the green arrows. (b) Under the spherical density-equalizing map, the region with a higher density (in yellow) expands and the region with a lower density (in purple) shrinks. }
    \label{fig:sphericalDEM_illustration}
\end{figure}

Considering the diffusion of $\rho$, we have the following diffusion equation:
\begin{equation}\label{eqt:diffusion}
\frac{\partial \rho}{\partial t} = \Delta \rho. 
\end{equation}
The density gradient induces a velocity field for the vertices (Fig.~\ref{fig:sphericalDEM_illustration}(a)):
\begin{equation}\label{eqt:velocity}
\mathbf{v}(\mathbf{x}(t),t) = - \frac{\nabla \rho(\mathbf{x}(t),t)}{\rho(\mathbf{x}(t),t)}.
\end{equation}
Now, note that the velocity field is determined by the prescribed populations and does not necessarily lie on the sphere. To ensure that the vertices remain on the sphere under the diffusion process, we perform a projection of $\mathbf{v}$ (Fig.~\ref{fig:sphericalDEM_illustration}(a)) as follows. First, we define the normal component of $\mathbf{v}(\mathbf{x}(t),t)$ at a point $\mathbf{x}(t) \in \mathbb{S}^2$ as
\begin{equation}
    \mathbf{v}^{\perp}(\mathbf{x}(t),t) = \left(\mathbf{v}(\mathbf{x}(t),t) \cdot \mathbf{n}(\mathbf{x}(t))\right) \mathbf{n}(\mathbf{x}(t)),
\end{equation}
where $\mathbf{n}$ is the outward unit normal on the sphere. Then, the projection of $\mathbf{v}$ onto the sphere is given by
\begin{equation}\label{eqt:velocity_projected}
    \widetilde{\mathbf{v}}(\mathbf{x},t) = \mathbf{v}(\mathbf{x},t) - \mathbf{v}^{\perp}(\mathbf{x},t).
\end{equation}
Now, every vertex on the sphere is displaced according to the following equation: 
\begin{equation}\label{eqt:displacement}
\mathbf{x}(t) = \mathbf{x}(0) + \int_0^t \widetilde{\mathbf{v}}(\mathbf{x}(\tau),\tau) d \tau.
\end{equation}
As $t \to \infty$, we obtain a spherical density-equalizing map with the final position of all vertices being $\mathbf{x}(\infty)$ (Fig.~\ref{fig:sphericalDEM_illustration}(b)). 

In the discrete case, we begin with the initial density $\rho^0_{\mathcal{F}} = \frac{\text{Population}}{\text{Triangle Area}}$ defined on each triangle element of the spherical mesh. The diffusion-based deformation is obtained by iteratively updating the positions of all vertices $\{\mathbf{x}_i(t_n)\}_i$ at time $t = t_n$, where $i = 1, 2, \cdots, |\mathcal{V}|$ and $n = 0, 1, \dots$, with a uniform timestep $\delta t = t_{n+1} - t_n$. Denote $\mathbf{v}^n_{\mathcal{V}}$ as a $|\mathcal{V}| \times 3$ matrix representing the velocity field at all vertices $\{\mathbf{x}_i(t_n)\}_i$ at time $t = t_n$. We also denote $\rho^n_{\mathcal{V}}$ and $\rho^n_{\mathcal{F}}$ respectively as a $|\mathcal{V}| \times 1$ column vector and a $|\mathcal{F}| \times 1$ column vector representing the density at every vertex or every face at time $t = t_n$. As described in Ref.~\cite{choi2018density}, $\rho^n_{\mathcal{F}}$ and $\rho^n_{\mathcal{V}}$ are related by a simple formula:
\begin{equation}\label{eqt:conversion}
    \rho^n_{\mathcal{V}} = M \rho^n_{\mathcal{F}},
\end{equation}
where $M$ is a $|\mathcal{V}| \times |\mathcal{F}|$ face-to-vertex conversion matrix with 
\begin{equation}\displaystyle
    M_{ij} =  \left\{\begin{array}{ll}
\frac{\text{Area}(T_j)}{\sum_{T \in N^{\mathcal{F}}(\mathbf{x}_i)}\text{Area}(T)} & \text{ if $T_j$ is incident to $\mathbf{x}_i$},\\
0 & \text{ otherwise.}
\end{array}\right.
\end{equation}
Here, $N^{\mathcal{F}}(\mathbf{x}_i)$ represents the set of all triangular faces incident to the vertex $\mathbf{x}_i$. Alternatively, one may also directly prescribe a vertex-based density $\rho^n_{\mathcal{V}}$ and skip Eq.~\eqref{eqt:conversion}.

To solve Eq.~\eqref{eqt:diffusion}, we discretize the Laplace--Beltrami operator at time $t = t_n$ as:
\begin{equation}
    \Delta_n = -A_{n}^{-1} L_n,
\end{equation}
where $A_{n}$ is a $|\mathcal{V}| \times |\mathcal{V}|$ diagonal matrix (called the lumped mass matrix) with 
\begin{equation}\label{eqt:A_n}
(A_{n})_{ii} = \frac{1}{3} \sum_{[\mathbf{x}_i,\mathbf{x}_j,\mathbf{x}_k] \in N^{\mathcal{F}}(\mathbf{x}_i)} \text{Area}([\mathbf{x}_i(t_n),\mathbf{x}_j(t_n),\mathbf{x}_k(t_n)]),
\end{equation}
where $[\mathbf{x}_i,\mathbf{x}_j,\mathbf{x}_k]$ is a triangle in $N^{\mathcal{F}}(\mathbf{x}_i)$. In other words, each diagonal entry of $A_{n}$ is the area sum of all triangles surrounding a vertex at time $t = t_n$. $L_{n}$ is the $|\mathcal{V}| \times |\mathcal{V}|$ mesh Laplacian matrix given by the cotangent formula~\cite{pinkall1993computing}:
\begin{equation}\label{eqt:L_n}
\displaystyle
(L_{n})_{ij} = \left\{\begin{array}{ll}
- \frac{1}{2}(\cot \alpha_{ij} + \cot \beta_{ij}) & \text{ if } [\mathbf{x}_i,\mathbf{x}_j] \in \mathcal{E},\\
- \sum_{k \neq i} (L_n)_{ik} & \text{ if } j = i,\\
0 & \text{ otherwise,}
\end{array}\right.
\end{equation}
where $\alpha_{ij}$ and $\beta_{ij}$ are the two angles opposite to the edge $[\mathbf{x}_i,\mathbf{x}_j]$ at time $t = t_n$. Eq.~\eqref{eqt:diffusion} can then be expressed as 
\begin{equation}\label{eqt:diffusion_discrete}
    \rho^{n+1}_{\mathcal{V}} = (A_{n} + \delta t L_{n})^{-1} (A_n \rho^{n}_{\mathcal{V}}).
\end{equation}
After solving Eq.~\eqref{eqt:diffusion_discrete}, we obtain the velocity field using Eq.~\eqref{eqt:velocity}. A discretization of the gradient $\nabla \rho$ for every triangular face $T = [\mathbf{x}_i,\mathbf{x}_j,\mathbf{x}_k] \in \mathcal{F}$ at time $t_n$ is given as follows (see Ref.~\cite{choi2018density} for details):
\begin{equation}
\nabla \rho^{n}_{\mathcal{F}}(T)  = \frac{\mathbf{n}^n_{\mathcal{F}}(T) \times \left(\rho^{n}_{\mathcal{V}}(\mathbf{x}_i) e_{jk} + \rho^{n}_{\mathcal{V}}(\mathbf{x}_j) e_{ki} + \rho^{n}_{\mathcal{V}}(\mathbf{x}_k) e_{ij}\right)}{2\text{Area}(T)},
\end{equation}
where $e_{ij}, e_{jk}, e_{ki}$ are the three directed edges $[\mathbf{x}_i(t_n),\mathbf{x}_j(t_n)]$, $[\mathbf{x}_j(t_n),\mathbf{x}_k(t_n)]$, $[\mathbf{x}_k(t_n),\mathbf{x}_i(t_n)]$, and $\mathbf{n}^n_{\mathcal{F}}(T)$ is the outward unit normal of the triangular face $T$. We can then easily convert $\nabla \rho^{n}_{\mathcal{F}}$ to $\nabla \rho^{n}_{\mathcal{V}}$ using the above-mentioned face-to-vertex conversion matrix $M$: 
\begin{equation}\label{eqt:gradrho_conversion}
    \nabla \rho^{n}_{\mathcal{V}} = M \nabla \rho^{n}_{\mathcal{F}}.
\end{equation}
Now, Eq.~\eqref{eqt:velocity} becomes
\begin{equation}\label{eqt:velocity_discrete}
    \mathbf{v}^n_{\mathcal{V}}(\mathbf{x}_i) = -\frac{\nabla \rho^{n}_{\mathcal{V}}(\mathbf{x}_i)}{\rho^{n}_{\mathcal{V}}(\mathbf{x}_i)}.
\end{equation}
The discretization of Eq.~\eqref{eqt:velocity_projected} then follows naturally:
\begin{equation} \label{eqt:velocity_projected_discrete}
\widetilde{\mathbf{v}}^n_{\mathcal{V}}(\mathbf{x}_i) = \mathbf{v}^n_{\mathcal{V}}(\mathbf{x}_i) - (\mathbf{v}^n_{\mathcal{V}}(\mathbf{x}_i) \cdot \mathbf{n}_{\mathcal{V}}(\mathbf{x}_i)) \mathbf{n}_{\mathcal{V}}(\mathbf{x}_i),
\end{equation}
where $\mathbf{n}_{\mathcal{V}}$ is the outward unit normal at the vertices. We can then update the position of all vertices $\mathbf{x}_i$ based on Eq.~\eqref{eqt:displacement}:
\begin{equation} \label{eqt:displacement_discrete}
\mathbf{x}_i(t_{n+1}) = \mathbf{x}_i(t_{n}) - \delta t \ \widetilde{\mathbf{v}}^n_{\mathcal{V}}(\mathbf{x}_i).
\end{equation}
While the projection of the velocity effectively eliminates the normal component, there may still be small numerical errors that cause $\mathbf{x}_i$ to move out of the spherical surface. This can be easily corrected by dividing $\mathbf{x}_i$ by its Euclidean 2-norm $\|\mathbf{x}_i\|_2$. Finally, as the number of iterations $n$ increases, the density is equalized over the entire sphere, i.e., all values of $\rho^{n}_{\mathcal{V}}$ converge to their average. 

\subsubsection{Enforcing the bijectivity of the mapping throughout the iterative process}\label{sect:sdem_overlap}
As discussed in~\cite{lyu2023bijective}, the bijectivity of the traditional density-equalizing mapping process is not guaranteed. Specifically, mesh overlaps may occur in the mapping result if the vertex displacement is too large. Analogously, in the spherical density-equalizing mapping process proposed in the above section, Eq.~\eqref{eqt:displacement_discrete} may cause some undesirable mesh overlaps under extremely large density gradients.

To resolve this issue, we propose an overlap correction scheme as follows. The main idea behind this scheme, as implied by Theorem~\ref{qc_bijective}, is to utilize quasi-conformal theory to rectify the fold-overs caused by the deformation in Eq.~\eqref{eqt:displacement_discrete}. For simplicity, we denote the initial conformal sphere as $\mathbb S^2_0$ and the resulting deformed sphere at the $n$-th iteration as $\mathbb S^2_n$. To correct the overlaps on $\mathbb S^2_n$, we first apply the stereographic projection to map both $\mathbb S^2_n$ and $\mathbb S^2_0$ onto the plane, resulting in two corresponding planar domains referred to as $D_n$ and $D_0$, respectively. 

Note that the stereographic projection requires choosing a specific part of the sphere as the north pole. While the projection is conformal in theory, the conformality distortion of the triangle elements under the projection may be affected by the choice of the pole. Here, to reduce the conformality distortion of the projection, we select the triangle with the most regular 1-ring neighborhood as the north pole. More specifically, we first define the face regularity $R^i_{\mathcal{F}}$ for each triangle face $T_i$ as follows:
\begin{equation}
R^i_{\mathcal{F}} = \sum^3_{j = 1} \left| \frac{e^i_j}{e^i_1 + e^i_2 + e^i_3} - \frac{1}{3} \right|,
\end{equation}
where $e^i_1, e^i_2, e^i_3$ represent the length of the three edges of $T_i$. Then, the vertex regularity can be defined by:
\begin{equation}
    R_{\mathcal{V}} = H^{\mathcal{F}\mathcal{V}} R_{\mathcal{F}},
\end{equation}
Here $R_{\mathcal{V}}$ is a $|\mathcal{V}|\times 1$ matrix with the $j$-th entry $R_{\mathcal{V}}^j$ being the vertex regularity of the $j$-th vertex , $R_{\mathcal{F}}$ is a $|\mathcal{F}|\times 1$ face regularity matrix, and $H^{\mathcal{F}\mathcal{V}}$ is a $|\mathcal{V}|\times |\mathcal{F}|$ sparse matrix such that
\begin{equation}\displaystyle
    H^{\mathcal{F}\mathcal{V}}_{ij} =  \left\{\begin{array}{ll}
    \dfrac{1}{3} & \text{ if $T_j$ contains the $i$-th vertex},\\
0 & \text{ otherwise.}
\end{array}\right.
\end{equation}
Finally, we determine the 1-ring face regularity $\overline{R}^i$ for each triangle $T_i$ by calculating the average vertex regularity of the vertices in its 1-ring neighborhood:
\begin{equation}
\overline{R}^i = \frac{1}{3} \sum_{\mathbf{x}_j \in T_i} R^j_{\mathcal{V}}.
\end{equation}
The most regular 1-ring triangular face is then chosen as the one with the lowest value of $\overline{R}^i$. 

After applying the stereographic projection to $\mathbb S^2_0$ and $\mathbb S^2_n$ using the above-mentioned approach, we can compute the Beltrami coefficient $\mu_n$ of the mapping between $D_0$ and $D_n$. By Theorem~\ref{qc_bijective}, the mapping is bijective if and only if the supremum norm of its Beltrami coefficient is less than 1. Therefore, the presence of mesh fold-overs can be determined by evaluating the norm of the Beltrami coefficient. In order to correct the overlaps, we introduce a truncation function $\mathbb{I}$ for the Beltrami coefficient: 
\begin{equation}
\mathbb I(\mu_n)= \begin{cases} \mu_n & \text { if } |\mu_n|< 1, \\ (1-\delta) \frac{\mu_n}{|\mu_n|} & \text { if } |\mu_n| \geq 1,\end{cases}
\end{equation}
where $\delta$ is a prescribed truncation parameter. In practice, we set $\delta = 0.1$. For now, we neglect the outermost faces and only apply the truncation function to the faces at the central region of $D_0$. Denote the truncated Beltrami coefficient by $\widetilde{\mu}_n$. We can then obtain a new deforming map $\widetilde{g}_n$ by applying the LBS algorithm~\cite{lui2013texture}:
\begin{equation}
\widetilde{g}_n = \textbf{LBS}(\widetilde{\mu}_n),
\end{equation}
with the outermost region fixed. The updated spherical parameterization can be obtained through the inverse stereographic projection. Note that in some rare cases with a large number of mesh overlaps, the above truncation procedure may give a highly discontinuous Beltrami coefficient and hence still lead to a non-bijective reconstructed map. In such cases, one may halve the step size in Eq.~\eqref{eqt:displacement_discrete} and repeat the above correction procedure.

Now, note that the faces near the north pole were neglected in the above procedure. To also address potential overlaps of those faces, we project the updated spherical parameterization onto the complex plane using the south-pole stereographic projection. Then, we can repeat the above procedure to correct the mesh overlaps using the Beltrami coefficient, and finally obtain a bijective spherical parameterization using the inverse south-pole stereographic projection.

By including this north pole–south pole overlap correction scheme at each iteration of the spherical density-equalizing mapping process, we can ensure that the bijectivity is preserved throughout the iterative process.

\subsubsection{Re-coupling the deformation and density throughout the iterative process}\label{sect:sdem_recoupling}
Recall that the proposed spherical density-equalizing mapping method uses density diffusion to produce shape deformations on the sphere. In particular, the density values and density gradients are converted between vertices and faces as described in Eq.~\eqref{eqt:conversion} and Eq.~\eqref{eqt:gradrho_conversion} at each iteration, and the vertex positions are updated accordingly. However, in the discrete case, such a conversion may result in a smoothing effect and hence numerical inaccuracies in the density may accumulate throughout the iterative process. Also, the vertex positions may be changed under the overlap correction scheme. Consequently, the shape deformation at each iteration may not perfectly correspond to the actual density flow. To address this issue, we propose an additional step that re-couples the deformation and the density at the end of each iteration. 

Specifically, at the end of the $n$-th iteration, instead of directly passing the updated density obtained from Eq.~\eqref{eqt:diffusion_discrete} to the next iteration, we redefine the density on each triangular face $T = [\mathbf{x}_i,\mathbf{x}_j,\mathbf{x}_k] \in \mathcal{F}$ using the updated map $\mathbf{x}(t_{n+1})$ as follows:
\begin{equation}\label{eqt:coupling}
    \rho^{n+1}_{\mathcal{F}}(T) = \frac{\text{Population}}{\text{Area}\left([\mathbf{x}_i(t_{n+1}), \mathbf{x}_j(t_{n+1}), \mathbf{x}_k(t_{n+1})]\right)}.
\end{equation}
After getting the new density $\rho^{n+1}_{\mathcal{F}}$, we can update $n=n+1$ and proceed to Eq.~\eqref{eqt:conversion} for the next iteration. In other words, instead of solving one single density diffusion equation over time and obtaining the shape deformation at every iteration based on it, we solve a new density diffusion equation at every iteration based on the current mapping result and use the computed density gradient for one iteration only. With this additional re-coupling step, the accuracy of the final mapping result can be improved.

\subsubsection{Summary}
Putting together the initial spherical conformal parameterization (Section~\ref{sect:sdem_initial}) and the iterative scheme for density-equalization on the sphere (Section~\ref{sec:DEM_sphere}) with the overlap correction scheme (Section~\ref{sect:sdem_overlap}) and the re-coupling scheme (Section~\ref{sect:sdem_recoupling}), we have the proposed SDEM algorithm for computing spherical density-equalizing maps. 
Following~\cite{choi2018density}, we define the stopping criterion for the iterative scheme as $\displaystyle \frac{\text{sd}\left(\rho^{n}_{\mathcal{V}}\right)}{\text{mean}\left(\rho^{n}_{\mathcal{V}}\right)} < \epsilon$, where $\epsilon$ is a stopping parameter. The algorithm is summarized in Algorithm~\ref{alg:SDEM}. In practice, the time step size is set to be $\delta t = 0.1$, the stopping parameter is set to be $\epsilon = 10^{-3}$, and the maximum number of iterations is $n_{\text{max}} = 200$.

\begin{algorithm}[h]
\KwIn{A genus-0 closed surface $\mathcal{M}$, a prescribed population, an initial spherical conformal parameterization $f_0:\mathcal{M} \to \mathbb{S}^2$, the stopping parameter $\epsilon$, and the maximum number of iterations allowed $n_{\max}$.}
\KwOut{A spherical density-equalizing map $f:\mathcal{M}\to \mathbb{S}^2$.}
\BlankLine

Compute the initial density $\rho^0_{\mathcal{F}}$ on $f_0(\mathcal{M})$ based on the prescribed population\;

Set $n = 0$\;

Let $\mathbf{x}_i(0) = f_0(v_i)$  for every vertex $v_i \in \mathcal{V}$\;

\Repeat{$\frac{\text{sd}\left(\rho^{n}_{\mathcal{V}}\right)}{\text{mean}\left(\rho^{n}_{\mathcal{V}}\right)} < \epsilon$ \ or \ $n \geq n_{\max}$}{ 

Compute $A_n$ and $L_n$ using Eq.~\eqref{eqt:A_n} and Eq.~\eqref{eqt:L_n}\;
Obtain $\rho^{n+1}_{\mathcal{V}}$ by solving the diffusion equation~\eqref{eqt:diffusion_discrete}\;

Compute the velocity field ${\mathbf{v}}^{n+1}_{\mathcal{V}}$ using Eq.~\eqref{eqt:velocity_discrete}\;

Compute the projected velocity field $\widetilde{\mathbf{v}}^{n+1}_{\mathcal{V}}$ using Eq.~\eqref{eqt:velocity_projected_discrete}\;

Update the position of all $\mathbf{x}_i (t_{n+1})$ using Eq.~\eqref{eqt:displacement_discrete}\;

Apply the overlap correction scheme in Section~\ref{sect:sdem_overlap} to further update $\mathbf{x}_i(t_{n+1})$\;

Apply the re-coupling scheme in Section~\ref{sect:sdem_recoupling} to update $\rho^{n+1}_{\mathcal{F}}$ using Eq.~\eqref{eqt:coupling}\;

Update $n = n + 1$\;
}

The resulting spherical density-equalizing map is given by $f(v_i) = \mathbf{x}_i(t_n)$ for all $i$\;

\caption{Bijective spherical density-equalizing map (SDEM)}
\label{alg:SDEM}
\end{algorithm}

\subsection{Bijective landmark-aligned spherical density-equalizing map (LSDEM)}
In~\cite{lyu2023bijective}, it was shown that the planar density-equalizing mapping process can be reformulated as an energy minimization problem. Analogously, as the goal of our spherical density-equalizing mapping method is to transform the density gradient into shape deformations, the following energy is minimized throughout the mapping process:
\begin{equation}
E_{\text{SDEM}} = \int \left\|\nabla \rho \right\|^2.
\end{equation}

Moreover, it is worth noting that several prior works on spherical parameterization were also based on energy minimization. For instance, in~\cite{gu2004genus}, an iterative algorithm for computing spherical harmonic maps was developed by minimizing the harmonic energy:
\begin{equation}
    E_{\text{Harmonic}} =  \int \left\|\nabla f\right\|^2.
\end{equation}
As harmonic maps between genus-0 closed surfaces are conformal, minimizing $E_{\text{Harmonic}}$ gives a spherical conformal parameterization. Another iterative algorithm was proposed by Lui et al.~\cite{lui2007landmark} to achieve optimized landmark-aligned spherical parameterizations. More specifically, an additional landmark mismatch energy was incorporated into their algorithm:
\begin{equation}
    E_{\text{Landmark}} = \int  \left\|f(p_i) - q_i\right\|^2,
\end{equation}
where $\{p_i\}_{i=1}^k \leftrightarrow \{q_i\}_{i=1}^k$ are two sets of prescribed corresponding landmarks to be matched, with $k$ being the total number of landmarks. 

Motivated by the above approaches, here we aim to compute a bijective landmark-aligned spherical density-equalizing map (abbreviated as LSDEM) balancing density-equalization, conformality, and landmark-matching conditions. This is achieved by minimizing the following proposed combined energy: 
\begin{equation}\label{eqt:combined}
\begin{split}
    E &= \alpha \ E_{\text{SDEM}} + \beta  \ E_{\text{Harmonic}} + \gamma  \ E_{\text{Landmark}} \\
    &= \alpha \int  \left\|\nabla \rho \right\|^2 + \beta \int  \left\|\nabla f\right\|^2 + \gamma \int  \left\|f(p_i) - q_i\right\|^2,
\end{split}
\end{equation}
where $\alpha, \beta, \gamma$ are three nonnegative weighting parameters. Here, for the landmark-matching term $\gamma \int \left\|f(p_i) - q_i\right\|^2$, note that since all geodesic curves between points on the sphere $\mathbb S^2$ move along the great circle, the geodesic distance between two points decreases with their Euclidean distance. Therefore, we can simply use the standard Euclidean metric for the landmark-matching term to simplify our computation. 

\subsubsection{Descent direction for minimizing the proposed combined energy}\label{sect:lsdem_descent}

To compute the landmark-aligned spherical density-equalizing map, we need to find a suitable descent direction to minimize the proposed combined energy in Eq.~\eqref{eqt:combined}. In the following, we denote $ \displaystyle E_1 = \alpha \int \left\|\nabla \rho \right\|^2$, $\displaystyle E_2 = \beta \int \left\|\nabla f\right\|^2 $, $\displaystyle E_3 = \gamma \int \left\|f(p_i) - q_i\right\|^2$ for simplicity and determine the descent direction of them one by one.

We first consider the descent direction of $\displaystyle E_1 = \alpha \int \left\|\nabla \rho\right\|^2$. As explained in Section \ref{sec:SDEM}, the velocity field induced by the density is given by $\mathbf{v} = -\frac{\nabla \rho}{\rho}$. In order to maintain the deformed shape as a sphere, we remove the normal component $\mathbf{v}^{\perp} = (\mathbf{v} \cdot \mathbf{n})\mathbf{n}$ of the velocity $\mathbf{v}$. The descent direction can then be defined using the tangential velocity:
\begin{equation} 
    dE_1 = \alpha(\mathbf{v} - \mathbf{v}^{\perp}).
\end{equation}
We remark that in this minimization problem for $E_1$, we do not employ the usual variational method. This is because the descent direction obtained through the variational method can potentially lead to a local minimum. By contrast, the velocity field we use is based on Fick's law of diffusion and can achieve the density-equalizing result, which corresponds to the global minimum of $E_1$. 

For the second term $\displaystyle E_2 =\beta \int \left\|\nabla f \right\|^2$, the Euler--Lagrange equation can be derived as follows:
\begin{equation}
    \begin{aligned}
        \left. \frac{d}{dt} \right|_{t = 0}E_2(f+tg) & = \beta \int \left. \frac{d}{dt} \right|_{t = 0} \|\nabla (f + tg) \|^2 \\
        &= 2 \beta \int  \nabla f  \nabla g \\
        &= -2 \beta \int \Delta f g.
    \end{aligned}
\end{equation}
Note that the derivative in the above formula is negative when $g =  \beta \Delta f$. To ensure that the deformed map remains a sphere, we again remove the normal component $\Delta^{\perp} f = (\Delta f \cdot \mathbf{n})\mathbf{n}$ of $ \Delta f$ and only use its tangential component $\Delta \Bar{f} = \Delta f - \Delta^{\perp} f$. Now, $E_2$ achieves its minimum if and only if the tangential component $\Delta \Bar{f}$ vanishes, i.e., $\Delta f = \Delta^{\perp} f$. Therefore, we define
\begin{equation}
    dE_2 = \beta \Delta \Bar{f}.
\end{equation}

For the landmark mismatch energy $\displaystyle E_3 = \gamma \int \left\|f(p_i) - q_i\right\|^2$, we compute its Euler--Lagrange equation:
\begin{equation}
    \begin{aligned}
         \left. \frac{d}{dt} \right|_{t = 0} E_3(f+tg) & = \gamma \int \left. \frac{d}{dt} \right|_{t = 0} \left\|f(p_i)+tg(p_i) - q_i \right\|^2 \\
         &= 2 \gamma \int (f(p_i) - q_i) g(p_i).
    \end{aligned}
\end{equation}
Similar to the above, we remove the normal component of $\left. \frac{d}{dt} \right|_{t = 0} E_3(f+tg)$. The descent direction of $E_3$ can be expressed as follows:
\begin{equation}\label{eqt:land_descent}
    dE_3(p) = \left\{\begin{array}{ll}
         - \gamma (f(p_i)-q_i) & \text{ if } p = p_i, \ i=1,2,\cdots, k, \\
         0 &  \text{ otherwise.}
    \end{array}\right.
\end{equation}

Combining the above formulas, the descent direction for $E$ is 
\begin{equation}\label{eqt:descent}
    dE = dE_1 + dE_2 + dE_3.
\end{equation}
Hence, we can minimize the combined energy~\eqref{eqt:combined} by the following iterative scheme:
\begin{equation}
    f_{n+1} = f_{n} + \delta t \ dE^{n},
\end{equation}
where $\delta t$ is the time step size, $f_n$ is the map at the $n$-th iteration, and $dE^n$ is the descent direction for $E$ at the $n$-th iteration.

\subsubsection{Optimal rotation for further reducing the energy}\label{sect:lsdem_optimal_rotation}
While the iterative scheme described above can effectively reduce the proposed combined energy at every step, it is important to note that in practice, a large landmark mismatch can potentially result in mesh fold-overs during the deformation process. Specifically, a large landmark mismatch gives a large $dE_3$ and hence a large $dE$, which may subsequently lead to undesirable overlaps or intersections in the mesh and require extra effort to correct. Here, we introduce an extra step of optimal rotation to further reduce the landmark mismatch energy throughout the iterative process, thereby enhancing the overall quality of the mapping.

First, we establish the following result:
\begin{theorem}\label{lsdem_rotation_invariant}
    The energies $E_{\text{Harmonic}}$ and $E_{\text{SDEM}}$ are rotation invariant.
\end{theorem}
\begin{proof}
    For the energy $E_{\text{Harmonic}}$, we first denote the composition of any spherical mapping $g: \mathbb S^2 \rightarrow \mathbb S^2 $ with a rotation as $\overline{g}$ for simplicity. Our goal is to prove that the energy $E_{\text{Harmonic}}(\overline{g})$ remains the same as $E_{\text{Harmonic}}(g)$ for any rotation. Since any rotation on a sphere can be decomposed into rotations about the $x$-, $y$-, and $z$-axes, it suffices to prove the invariance of the energy under rotations about the three axes.

    Note that the gradient of a spherical mapping $g$ can be computed as $\nabla g = \left(\frac{\partial g}{\partial x}, \frac{\partial g}{\partial y}, \frac{\partial g}{\partial z}\right)$. The squared magnitude of the gradient is given by 
    \begin{equation}
    \left\|\nabla g\right\|^2 = \left(\frac{\partial g}{\partial x}\right)^2 + \left(\frac{\partial g}{\partial y}\right)^2 + \left(\frac{\partial g}{\partial z}\right)^2.
    \end{equation}
    Now, without loss of generality, we consider a rotation about the $z$-axis. The corresponding rotation matrix is given by
    \begin{equation}\label{rotation_matrix}
    R_z(\theta) = \begin{pmatrix}
         \cos{\theta} & -\sin{\theta} & 0\\ 
        \sin{\theta} & \cos{\theta}  & 0\\
        0 & 0 & 1\\
    \end{pmatrix},              
    \end{equation}
    where $\theta$ is the rotation angle. After applying this rotation, the point $(x, y, z)$ will be transformed to $(x^{R}, y^{R}, z^{R})$. We can then compute the gradient of the composition $\overline{g} = g \circ R_z(\theta)$ as follows:
    \begin{equation}
    \nabla \overline{g} = \left(\frac{\partial \overline{g}}{\partial x}, \frac{\partial \overline{g}}{\partial y},\frac{\partial \overline{g}}{\partial z}\right) = \left(\cos{\theta} \frac{\partial g}{\partial x^{R}} +\sin{\theta} \frac{\partial g}{\partial y^{R}}, \cos{\theta} \frac{\partial g}{\partial y^{R}} - \sin{\theta} \frac{\partial g}{\partial x^{R}}, \frac{\partial g}{\partial z^{R}}  \right),
    \end{equation} 
    and the squared magnitude of the gradient is given by:
    \begin{equation}
    \begin{aligned}
    \left\|\nabla \overline{g}\right\|^2  & = \left(\frac{\partial g}{\partial x^{R}}\right)^2 \left(\cos^2{\theta} + \sin^2{\theta}\right) + \left(\frac{\partial g}{\partial y^{R}}\right)^2 \left(\cos^2{\theta} + \sin^2{\theta}\right) + \left(\frac{\partial g}{\partial z^{R}}\right)^2 \\
    & = \left(\frac{\partial g}{\partial x^{R}}\right)^2 + \left(\frac{\partial g}{\partial y^{R}}\right)^2 + \left(\frac{\partial g}{\partial z^{R}}\right)^2.
    \end{aligned}
    \end{equation}
    Therefore, the harmonic energy after rotation is given by:
    \begin{equation}
    \begin{aligned}
            \int \left\|\nabla \overline{g} \right\|^2 &=\int \left( \left(\frac{\partial \overline{g}}{\partial x}\right)^2 + \left(\frac{\partial \overline{g}}{\partial y}\right)^2 +\left(\frac{\partial \overline{g}}{\partial z}\right)^2 \right) \\
            &=\int \left( \left(\frac{\partial g}{\partial x^{R}}\right)^2 + \left(\frac{\partial g}{\partial y^{R}}\right)^2 + \left(\frac{\partial g}{\partial z^{R}}\right)^2 \right) =\int \left\|\nabla g\right\|^2. 
    \end{aligned}
    \end{equation} 
    Thus, we conclude that $E_{\text{Harmonic}}$ remains unchanged under any rotation about the $z$-axis. Similarly, we can prove that $E_{\text{Harmonic}}$ is unchanged under any rotation about the $x$-axis or the $y$-axis. Consequently, $E_{\text{Harmonic}}$ is rotation invariant. 
    
    For $E_{\text{SDEM}}$, note that the density $\rho$ is a function from $\mathbb{S}^2$ to $\mathbb{R}$. Analogous to the above, it suffices to prove the invariance of the energy $E_{\text{SDEM}}$ under rotations about the three axes. 
    
    Again, without loss of generality, we consider a rotation about the $z$-axis and denote the composition of the density function with the rotation as $\overline{\rho} = \rho \circ R_z(\theta)$, where $R_z(\theta)$ is the rotation matrix defined in Eq.~\eqref{rotation_matrix}. Let $\left(x^{R},y^{R},z^{R} \right)$ represent the deformed position of $(x,y,z)$ after the rotation. Using the same method as above, the squared magnitude of the gradient integral is given by:
    \begin{equation}
    \begin{aligned}
           \int \left\|\nabla \overline{\rho}\right\|^2  &= \int \left(\left(\frac{\partial \overline{\rho}}{\partial x}\right)^2 + \left(\frac{\partial \overline{\rho}}{\partial y}\right)^2 +\left(\frac{\partial \overline{\rho}}{\partial z}\right)^2 \right)\\
           & = \int \left( \left(\frac{\partial \rho}{\partial x^{R}}\right)^2 + \left(\frac{\partial \rho}{\partial y^{R}}\right)^2 + \left(\frac{\partial \rho}{\partial z^{R}}\right)^2 \right) = \int \left\|\nabla \rho\right\|^2. 
    \end{aligned}
    \end{equation}
    Therefore, $E_{\text{SDEM}}$ is invariant under any rotation about the $z$-axis. Similarly, we can prove that $E_{\text{SDEM}}$ is unchanged under any rotation about the $x$-axis and $y$-axis. Consequently, $E_{\text{SDEM}}$ is also rotation invariant.\hfill$\blacksquare$

\end{proof}

By Theorem~\ref{lsdem_rotation_invariant}, the first two terms in the proposed combined energy $E$ in Eq.~\eqref{eqt:combined} are invariant under rotations. This allows us to consider finding an optimal rotation of the sphere to reduce the landmark mismatch energy $E_3$ without affecting the two other energy terms in the combined energy $E$. 

More specifically, let $\{f_n(p_i)\}_{i=1}^k$ be the current positions of the landmark vertices at the $n$-th iteration of the iterative scheme, and $\{q_i\}_{i=1}^k$ be their target positions on the sphere. We consider three rotation matrices $R_x, R_y, R_z$ about the $x$-axis, ,$y$-axis, and $z$-axis, respectively:
    \begin{equation}\label{rotation_matrix_xyz}
    R_x(\phi) = \begin{pmatrix}
         1 & 0 & 0\\
         0 & \cos{\phi} & -\sin{\phi}\\ 
        0 & \sin{\phi} & \cos{\phi}\\
    \end{pmatrix},  \ 
    R_y(\psi) = \begin{pmatrix}
         \cos{\psi} & 0 & \sin{\psi}\\ 
         0 & 1 & 0\\
        -\sin{\psi} & 0 & \cos{\psi} \\
    \end{pmatrix}, \ 
    R_z(\theta) = \begin{pmatrix}
         \cos{\theta} & -\sin{\theta} & 0\\ 
        \sin{\theta} & \cos{\theta}  & 0\\
        0 & 0 & 1\\
    \end{pmatrix},              
    \end{equation}
where $\phi, \psi, \theta$ are the rotation angles. We can then search for the optimal rotation angles that minimize the landmark mismatch error:
\begin{equation}
    L(\phi, \psi, \theta) = \sum_{i=1}^k \left\|R_x(\phi) R_y(\psi) R_z(\theta) f_n(p_i) - q_i \right\|^2.
\end{equation}
Specifically, it is easy to see that
\begin{equation}
    \frac{\partial L}{\partial \phi} = 2 \sum_{i=1}^k \left\langle R_x R_y R_z f_n(p_i) - q_i, \frac{\partial R_x}{\partial \phi} R_y R_z f_n(p_i) \right\rangle,
\end{equation}
\begin{equation}
    \frac{\partial L}{\partial \psi} = 2 \sum_{i=1}^k \left\langle R_x R_y R_z f_n(p_i) - q_i, R_x \frac{\partial R_y}{\partial \psi} R_z f_n(p_i) \right\rangle,
\end{equation}
\begin{equation}
    \frac{\partial L}{\partial \phi} = 2 \sum_{i=1}^k \left\langle R_x R_y R_z f_n(p_i) - q_i, R_x R_y \frac{\partial R_z}{\partial \theta} f_n(p_i) \right\rangle,
\end{equation}
and so one can easily obtain the optimal rotation angles $(\phi^*, \psi^*, \theta^*)$ using gradient descent. By applying the three optimal rotations $R_x(\phi^*)$, $R_y(\psi^*)$, $R_z(\theta^*)$ on the mapping $f_n$, the landmark mismatch energy $E_{\text{landmark}}$ of the updated spherical mapping $\widetilde{f}_n$ will be less than or equal to that of $f_n$ (see Fig.~\ref{fig:LSDEM_rotation} for an illustration). Moreover, since the two energies $E_{\text{Harmonic}}$ and $E_{\text{SDEM}}$ are unchanged under the rotation as proved in Theorem~\ref{lsdem_rotation_invariant}, it follows that the combined energy $E$ of $\widetilde{f}_n$ is also less than or equal to that of $f_n$. In other words, this extra optimal rotation step can help us further reduce the landmark mismatch before performing the descent step. We can then replace $f_n$ in Section~\ref{sect:lsdem_descent} with $\widetilde{f}_n$ and proceed with the descent step to obtain $f_{n+1}$.

\begin{figure}[t]
   \centering
   \includegraphics[width=\textwidth]{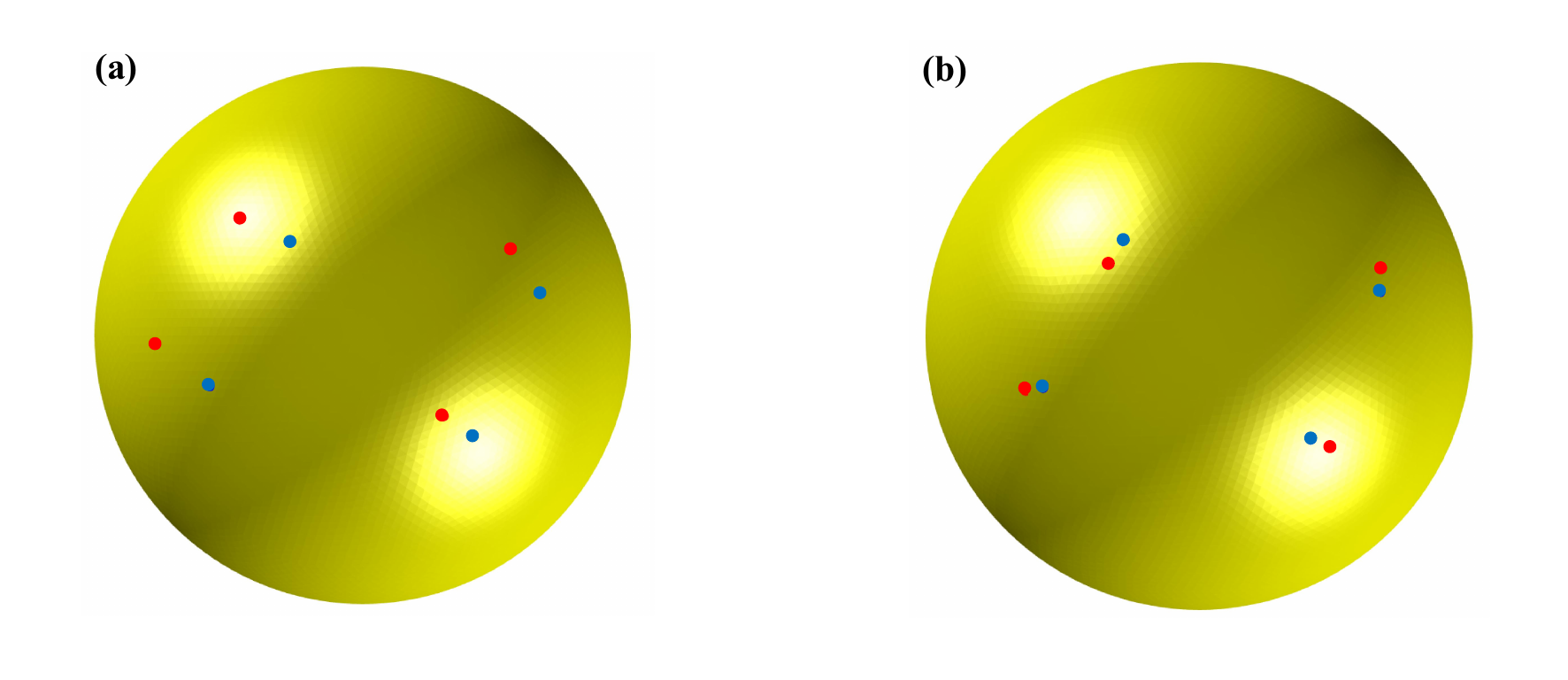}
    \caption{\textbf{An illustration of the optimal rotation step in the proposed LSDEM algorithm.} (a) A spherical mapping with the landmarks (red) and their target positions (blue). (b) The optimal rotation result, in which the landmark mismatch is reduced.}
    \label{fig:LSDEM_rotation}
\end{figure}

\subsubsection{Summary}
Analogous to the SDEM algorithm, here in the proposed LSDEM algorithm, we start by computing an initial spherical conformal parameterization (Section~\ref{sect:sdem_initial}), from which we can get the initial density and the initial combined energy. Then, we deform the initial spherical mapping using an iterative process consisting of both the optimal rotation step (Section~\ref{sect:lsdem_optimal_rotation}) and the descent step (Section~\ref{sect:lsdem_descent}) for minimizing the combined energy. The previously mentioned overlap correction scheme (Section~\ref{sect:sdem_overlap}) and the re-coupling scheme (Section~\ref{sect:sdem_recoupling}) are also included in the iterative process. We stop the iterations when the mapping result stabilizes: $\|f_n - f_{n-1}\| < \epsilon$, where $\epsilon$ is the stopping parameter. Ultimately, $f = f_n$ is the desired landmark-aligned bijective spherical density-equalizing map. 

The proposed LSDEM method is summarized in Algorithm~\ref{alg:LSDEM}. In practice, the time step size is set to be $\delta t = 0.01$, the stopping parameter is set to be $\epsilon = 10^{-3}$, and the maximum number of iterations is $n_{\text{max}} = 200$. 

\begin{algorithm}[h]
\KwIn{A genus-0 closed surface $\mathcal{M}$, a prescribed population, an initial spherical conformal parameterization $f_0:\mathcal{M} \to \mathbb{S}^2$,
the landmarks $\{p_1,\dots ,p_k\} \subset \mathcal{M}$ and target positions $\{q_1,\dots ,q_k\} \subset \mathbb{S}^2$, the stopping parameter $\epsilon$, and the maximum number of iterations allowed $n_{\max}$.}
\KwOut{A landmark-aligned spherical density-equalizing map $f:\mathcal{M}\to \mathbb{S}^2$.}
\BlankLine

Compute the initial density $\rho^0_{\mathcal{F}}$ on $f_0(\mathcal{M})$ based on the prescribed population\;

Set $n = 0$\;

\Repeat{$\|f_n - f_{n-1}\| < \epsilon$ \ or \ $n \geq n_{\max}$}{ 
Compute an optimal rotation of $f_n$ as described in Section~\ref{sect:lsdem_optimal_rotation} and obtain $\widetilde{f}_n$ \;

Compute the descent direction $dE^n$ by Eq.~\eqref{eqt:descent} and Update $f_{n+1} = \widetilde{f}_n + \delta t \ dE^n$ \;

Apply the overlap correction scheme in Section~\ref{sect:sdem_overlap} to further update $f_{n+1}$\;

Apply the re-coupling scheme in Section~\ref{sect:sdem_recoupling} to update $\rho^{n+1}_{\mathcal{F}}$ using Eq.~\eqref{eqt:coupling}\;

Update $n = n + 1$\;

}

The resulting landmark-aligned spherical density-equalizing map is given by $f = f_n$\;

\caption{Landmark-aligned bijective spherical density-equalizing map (LSDEM)}
\label{alg:LSDEM}
\end{algorithm}
 
We remark that by changing the values of the parameters $\alpha$, $\beta$, $\gamma$ in the combined energy $E$, we can achieve different desired mapping effects in the LSDEM result. Specifically, by using a larger value of $\alpha$, we can achieve a more density-equalizing result. By increasing the value of $\beta$, we can reduce the angle distortion in the mapping result. Lastly, by increasing the value of $\gamma$, we can reduce the landmark mismatch in the mapping result.


\begin{figure}[t]
    \centering
    \includegraphics[width=\textwidth]{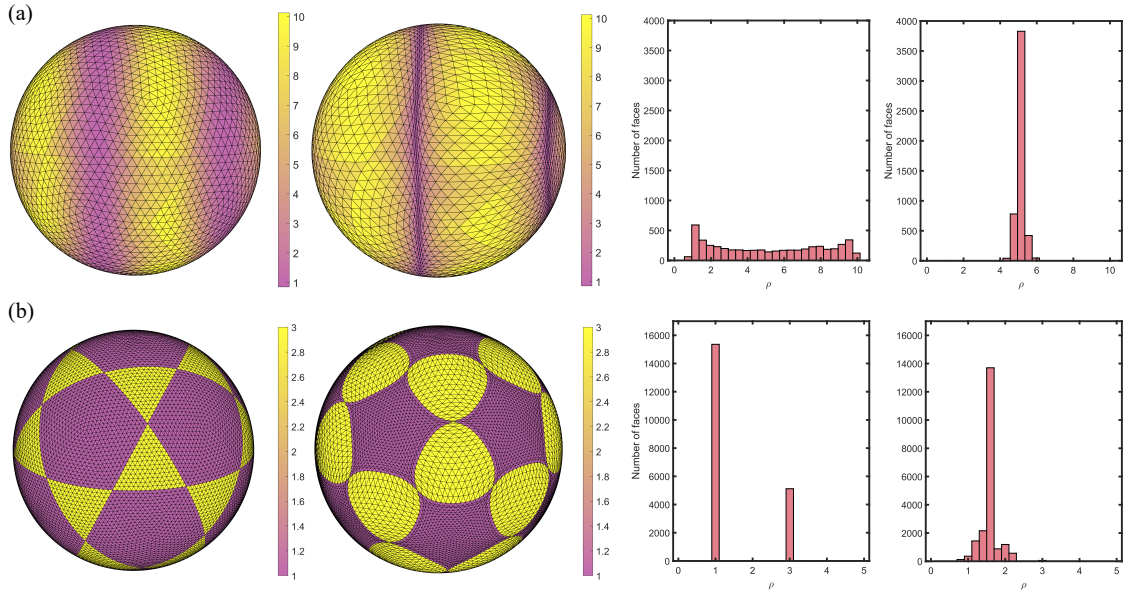}
    \caption{\textbf{Spherical density-equalizing maps of spherical surfaces.} Each row shows one example. (a) An example with continuous input density. (b) An example with discontinuous input density. Left to right: The initial spherical surface color-coded with the initial density, the final SDEM result color-coded with the initial density, the histogram of the initial density, and the histogram of the final density.}
    \label{fig:sdem_sphere}
\end{figure}

\section{Experiments}\label{sect:experiment}
In this section, we present experimental results to demonstrate the effectiveness of our proposed SDEM and LSDEM algorithms. The algorithms are implemented using MATLAB R2021a on the Windows platform. All experiments are conducted on a computer with an Intel(R) Core(TM) i9-12900 2.40 GHz processor and 32GB memory. The surface meshes are from online mesh repositories~\cite{common}. All surfaces are discretized in the form of triangular meshes.

\subsection{Spherical density-equalizing map}
We start by considering some examples of mapping a spherical surface using the proposed SDEM method. In the example shown in Fig.~\ref{fig:sdem_sphere}(a), we define different populations at different regions on the sphere to give a continuous initial density, with the maximum density and minimum density different by 10 times. We then apply the proposed method and obtain the spherical density-equalizing mapping result as shown in the figure, in which the high-density region is enlarged significantly and the low-density region is shrunk. By considering the initial and final density histograms, we can see that the density is effectively equalized using our method.

We then consider a more complicated spherical example involving multiple regions with discontinuous densities (Fig.~\ref{fig:sdem_sphere}(b)). In particular, we define different populations on the sphere such that the initial density of all the triangular regions is three times the density of all the pentagonal regions. It can be observed that all the triangular regions expand and all the pentagonal regions shrink under the spherical density-equalizing map. We can also see from the initial and final density histograms that the density is effectively equalized using our method.

\begin{figure}[t]
    \centering
    \includegraphics[width=\textwidth]{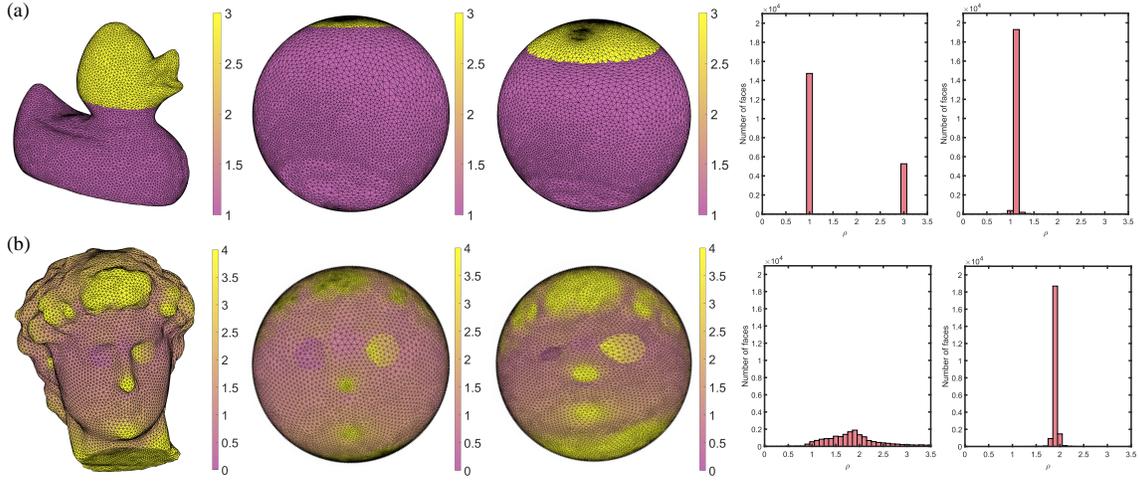}
    \caption{\textbf{Spherical density-equalizing maps of genus-0 closed surfaces.} Each row shows one example. (a) The Duck model. (b) The David model. Left to right: The input surface mesh color-coded with the initial density, the spherical conformal parameterization color-coded with the initial density, the final SDEM result color-coded with the initial density, the histogram of the initial density, and the histogram of the final density.}
    \label{fig:sdem_surface}
\end{figure}

Next, we consider computing spherical density-equalizing maps for general genus-0 surfaces. In Fig.~\ref{fig:sdem_surface}(a), we consider mapping the Duck model with the head part of it enlarged. As shown in the figure, we can apply our proposed SDEM method and compute a spherical density-equalizing map with the desired effect. The histograms show that the density is effectively equalized. Similarly, as shown in Fig.~\ref{fig:sdem_surface}(b), we can effectively magnify and shrink different parts in the mapping result of the David model. In particular, one of its eyes is enlarged significantly while the other one is shrunk.

For a more quantitative analysis of our proposed SDEM method, in Table~\ref{tab:SDEM} we record the computational time, variance of the initial density, variance of the final density, and number of overlaps for mapping different surface models. It can be observed that under our SDEM algorithm, the variance in the density can be significantly reduced. Also, from the number of overlaps, we can see that the mappings are all bijective. As for the computational time, we can see that our algorithm is highly efficient and only takes a few seconds for all examples.

\begin{table}[t!]
\small
    \caption{\textbf{The performance of our SDEM algorithm.} For each surface, we record the number of triangle elements, the computational time, the variance of the normalized initial density $\widetilde{\rho}_1 = \frac{\rho_1}{\text{Mean}(\rho_1)}$ and the normalized final density $\widetilde{\rho}_2  = \frac{\rho_2}{\text{Mean}(\rho_2)}$, where $\rho_1$ is the initial vertex density and $\rho_2$ is the final vertex density, and the number of overlaps.}\label{tab:SDEM}
  \begin{center}
  \begin{tabular}{|c|c|c|c|c|c|} \hline
    \bf Surface & \bf \# Faces & \bf Time (s) &\bf $\text{Var}(\widetilde{\rho}_1)$  &\bf $\text{Var}(\widetilde{\rho}_2)$ & \bf \# Overlaps \\\hline
    Sphere 1 (Fig.~\ref{fig:sdem_sphere}(a)) & 5120 & 0.4471 & 0.3236 & $< 0.0001$ & 0 \\ \hline
    Sphere 2 (Fig.~\ref{fig:sdem_sphere}(b)) & 20480 & 0.2443 & 0.2935 & 0.0015 & 0 \\ \hline
    Duck (Fig.~\ref{fig:sdem_surface}(a)) & 20000 & 1.4416  & 0.3294  & $< 0.0001$ & 0\\ \hline
    David (Fig.~\ref{fig:illustration}, enlarge) & 21338 & 1.4281 & 0.6099 & $< 0.0001$ & 0 \\ \hline
    David (Fig.~\ref{fig:illustration}, shrink) & 21338 & 1.4295 & 0.6004 & $< 0.0001$ & 0 \\ \hline
    David (Fig.~\ref{fig:sdem_surface}(b), mixed) & 21338 & 1.4267 & 0.5949  & $< 0.0001$ & 0 \\ \hline
  \end{tabular}
\end{center}
\end{table}

\begin{figure}[t!]
    \centering
    \includegraphics[width=\textwidth]{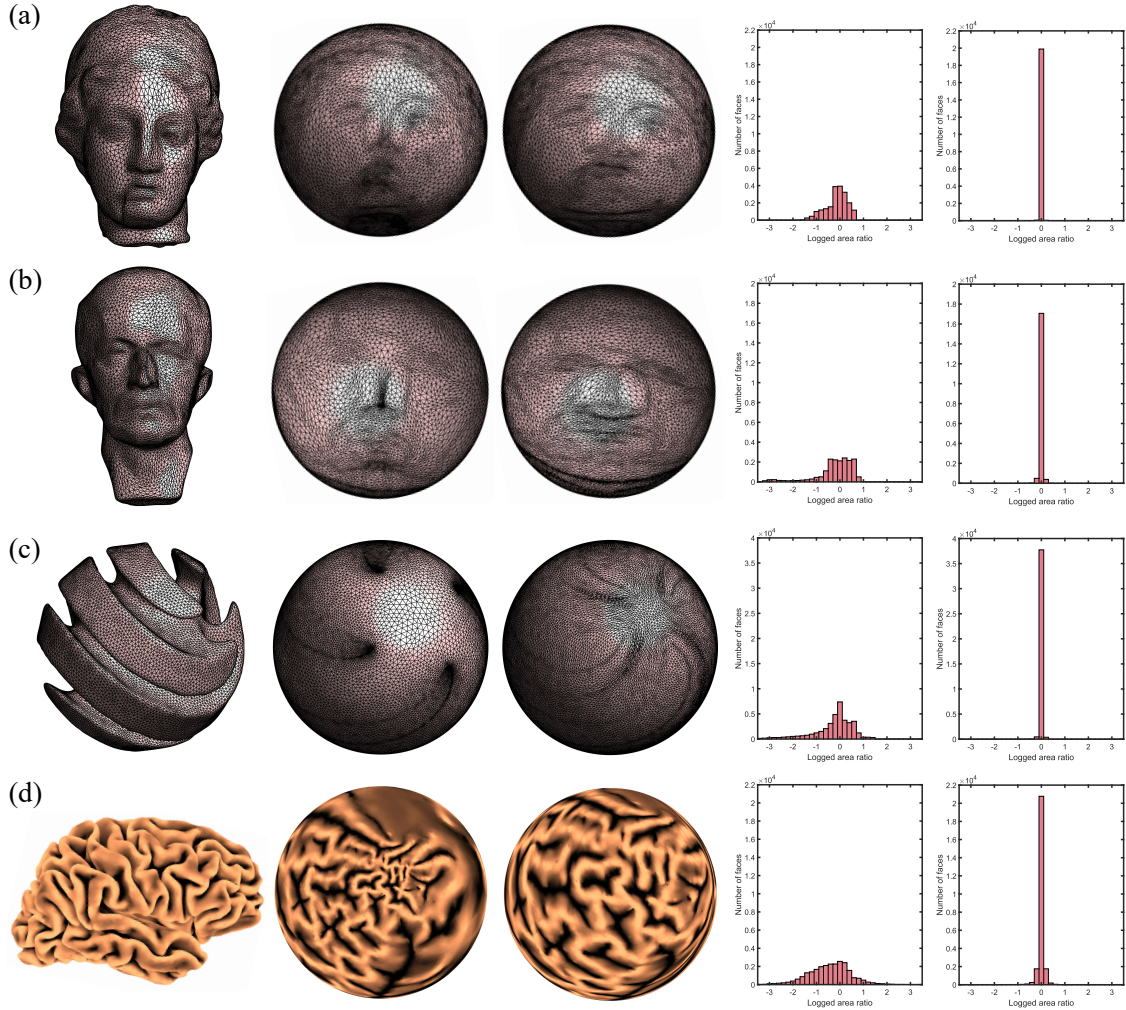}
    \caption{\textbf{Spherical area-preserving parameterization of genus-0 closed surfaces obtained by the SDEM method.} Each row shows one example. (a) The Igea model. (b) The Max Planck model. (c) A twisted ball. (d) A brain cortical surface. Left to right: The input surface mesh, the initial spherical conformal parameterization, the final SDEM result, the histogram of the logged area ratio $d_{\text{area}}$ of the initial spherical parameterization, and the histogram of the logged area ratio $d_{\text{area}}$ of the final spherical parameterization. The surfaces in (d) are color-coded with the brain surface curvature for better visualization.}
    \label{fig:sdem_area}
\end{figure}

Note that by setting the population as the face area of the input mesh and applying the SDEM method, we can achieve a spherical area-preserving parameterization. To illustrate this, we consider several genus-0 closed surfaces as shown in Fig.~\ref{fig:sdem_area}. For each surface, we first parameterize it onto the sphere using the spherical conformal parameterization. Then, we apply the SDEM method to obtain an area-preserving parameterization. To assess the area-preserving quality of the spherical parameterizations, we can evaluate the area distortion by considering the logged area ratio for every triangular face $T$:
\begin{equation}
    d_{\text{area}}(f)(T) = \log\left(\frac{\text{Area}(f(T))/\sum_{T'\in \mathcal{F}} \text{Area}(f(T'))}{\text{Area}(T)/\sum_{T'\in \mathcal{F}} \text{Area}(T')}\right),
\end{equation}
where $f$ is the spherical parameterization and $\mathcal{F}$ is the set of all triangular faces. Note that the two factors in the numerator and denominator are used for normalization so that a perfectly area-preserving parameterization would yield $d_{\text{area}} \equiv 0$. From the results in Fig.~\ref{fig:sdem_area}, it can be observed that our SDEM method can handle different genus-0 closed surfaces with highly complex geometry very well. Specifically, the initial area distortion histograms in Fig.~\ref{fig:sdem_area} show that the area distortion is very large in the initial conformal parameterization for all examples. By contrast, the final area distortion histograms are all highly concentrated, which indicates that the final mappings are highly area-preserving. For a more quantitative comparison, Table~\ref{tab:SDEM_area} records the computational time, initial and final area distortions, and the number of overlaps for different surfaces. In particular, it can be observed that our SDEM algorithm effectively reduces the area distortion by over 90\% in all experiments while maintaining the bijectivity of the mappings. This shows that our algorithm is capable of computing the spherical area-preserving parameterization for a wide range of genus-0 closed surfaces.

\begin{table}[t!]
\small
    \caption{\textbf{The performance of our SDEM algorithm for spherical area-preserving parameterization.} For each surface, we record the number of triangle elements, the computational time, mean and standard deviation of the initial area distortion $|d_{\text{area}}(f_0)|$ and the final area distortion $|d_{\text{area}}(f)|$, and the number of overlaps.}\label{tab:SDEM_area}
  \begin{center}
  \begin{tabular}{|C{15mm}|c|c|c|c|c|c|c|} \hline
    \multirow{ 2}{*}{\bf Surface} & \multirow{ 2}{*}{\bf \# Faces} & \multirow{ 2}{*}{\bf Time (s)}  & \multicolumn{2}{c|}{\bf $|d_{\text{area}}(f_0)|$} & \multicolumn{2}{c|}{\bf $|d_{\text{area}}(f)|$}  & \multirow{ 2}{*}{\bf \# Overlaps} \\ \cline{4-7}
    & & & \bf Mean & \bf SD & \bf Mean & \bf SD & \\ \hline
    Chinese Lion (Fig.~\ref{fig:illustration}) & 10000 & 6.2777 & 1.6837 & 1.6941 & 0.1199 & 0.1786 & 0 \\  \hline
    Igea (Fig.~\ref{fig:sdem_area}(a)) & 20000  & 1.1320 & 0.3668 &  0.3163 &  0.0146 & 0.0183 & 0 \\ \hline 
    Max Planck (Fig.~\ref{fig:sdem_area}(b)) & 18000 & 1.0756 & 0.5619 & 0.6252 & 0.0296 & 0.0411  & 0 \\ \hline 
    Twisted Ball (Fig.~\ref{fig:sdem_area}(c)) & 38620 & 1.5043 & 0.6240 & 0.7174 & 0.0210 & 0.0294 & 0 \\ \hline 
    Brain (Fig.~\ref{fig:sdem_area}(d)) & 25000 & 4.5765 & 0.7514 & 0.6222 & 0.0698 & 0.1414 & 0 \\ \hline
    Frog (Fig.~\ref{fig:remeshing_comparison}) & 20000 & 2.9655 & 0.9056 & 1.3382 & 0.0460 & 0.0822 & 0\\ \hline
  \end{tabular}
\end{center}
\end{table}

\begin{figure}[t!]
    \centering
    \includegraphics[width=0.8\textwidth]{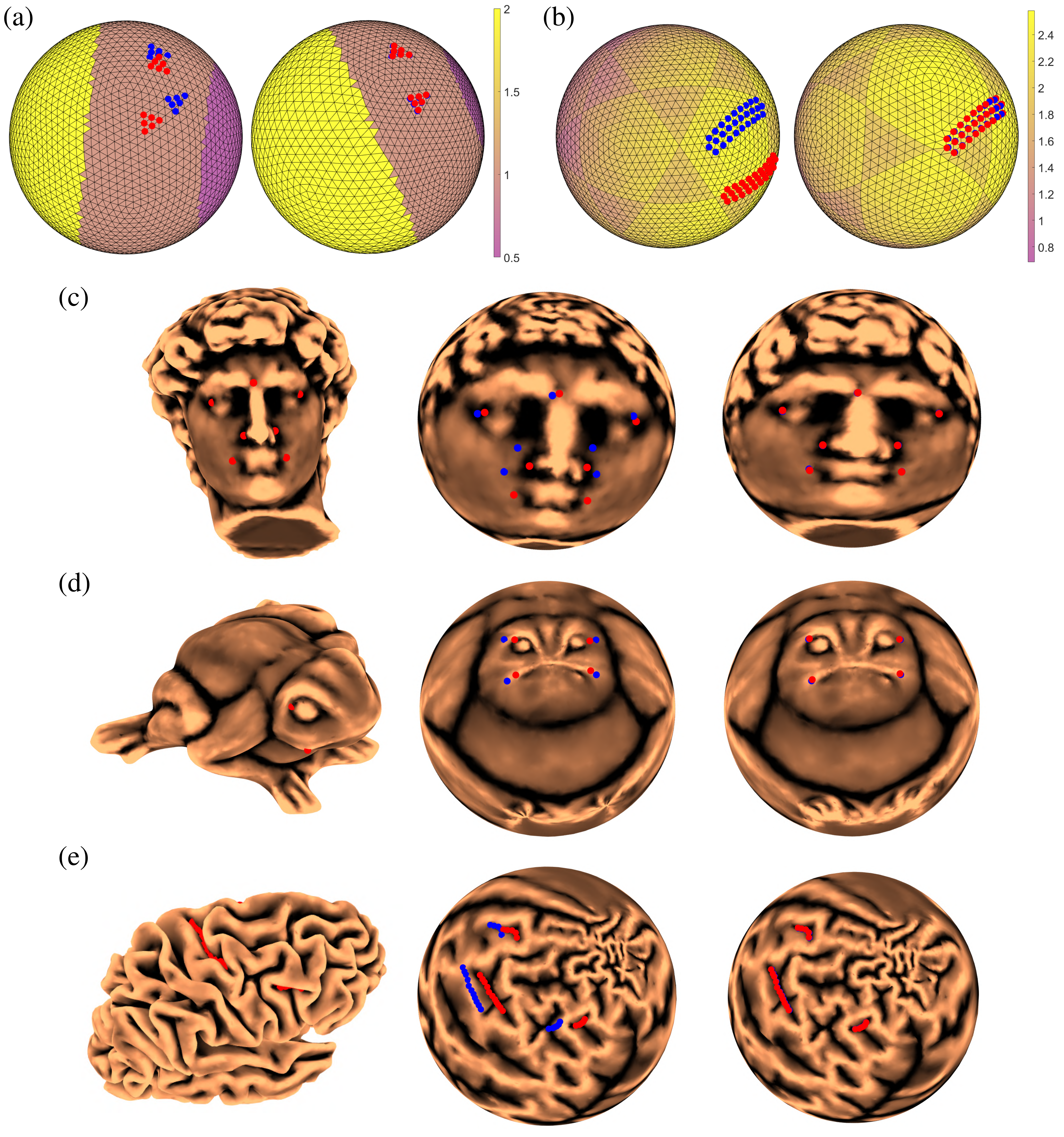}
    \caption{\textbf{Landmark-aligned spherical density-equalizing maps of genus-0 closed surfaces.} (a)--(b)~Two examples of mapping a spherical surface with different prescribed populations and landmark constraints. In each example, the left panel shows the initial spherical surface color-coded with the initial density together with the labelled landmarks (red dots) and the prescribed target positions (blue dots), and the right panel shows the final LSDEM result. (c)--(e) Examples of mapping general genus-0 closed surfaces. Left to right: The input genus-0 surface color-coded with the surface mean curvature and the labelled landmarks (red dots), the initial spherical conformal parameterization together with the labelled landmarks (red dots) and the prescribed target positions (blue dots), and the final LSDEM result. In all examples, the parameters in the energy~\eqref{eqt:combined} are set as $(\alpha, \beta, \gamma) = (1, 2, 5)$.}
    \label{fig:result_lsdem}
\end{figure}

\subsection{Landmark-aligned spherical density-equalizing map}
After demonstrating the effectiveness of our proposed SDEM method, we consider our proposed LSDEM method. Fig.~\ref{fig:result_lsdem}(a)--(b) show two examples of mapping a spherical surface with landmark constraints using the proposed LSDEM method, and Fig.~\ref{fig:result_lsdem}(c)--(e) shows several examples of mapping general genus-0 closed surfaces with prescribed landmark constraints. It can be observed in all examples that by using the proposed LSDEM method, we can obtain spherical mapping results with different regions enlarged or shrunk suitably and with all landmarks being well aligned.

\begin{table}[t!]
\small
    \caption{\textbf{The performance of our LSDEM algorithm.} For each surface, we record the number of triangle elements, the computational time, the variance of the normalized initial density $\widetilde{\rho}_1 = \frac{\rho_1}{\text{Mean}(\rho_1)}$ and the normalized final density $\widetilde{\rho}_2 = \frac{\rho_2}{\text{Mean}(\rho_2)}$, where $\rho_1$ is the initial vertex density and $\rho_2$ is the final vertex density, the mean of the norm of the Beltrami coefficient $|\mu|$, the 2-norm of the landmark error in initial surface $\textbf{LE}_1$, the 2-norm of the landmark error in the final result $\textbf{LE}_2$, and the number of overlaps. In all examples, the parameters in the energy~\eqref{eqt:combined} are set as $(\alpha, \beta, \gamma) = (1, 2, 5)$.}\label{tab:LSDEM}
  \begin{center}
  \setlength{\tabcolsep}{1mm}{
  \begin{tabular}{|C{13mm}|c|c|c|c|c|c|c|c|} \hline
    \bf Surface & \bf \# Faces & \bf Time (s) &\bf $\text{Var}(\widetilde{\rho}_1)$  &\bf $\text{Var}(\widetilde{\rho}_2)$  & \bf Mean$(|\mu|)$ & \bf $\textbf{LE}_1$ & $\textbf{LE}_2$  & \bf \# Overlaps \\\hline
    Sphere (triangle) (Fig.~\ref{fig:result_lsdem}(a)) & 5120 & 1.6780 & 0.1857 & 0.0077 & 0.0788 & 2.7450 & 0.0897 & 0 \\ \hline
    Sphere (rectangle) (Fig.~\ref{fig:result_lsdem}(b))  & 5120 & 1.2989 & 0.0873 & 0.0013 & 0.0479 & 15.9647  & 0.2501 & 0 \\ \hline
    David (Fig.~\ref{fig:result_lsdem}(c)) & 21338 & 13.1052 & 0.5962 & 0.0118  & 0.1470& 0.9222 & 0.0405 & 0 \\ \hline 
    Frog (Fig.~\ref{fig:result_lsdem}(d)) & 20000 & 10.8065 & 26.6130 &  0.6132  & 0.1411 & 0.2913 &  0.0459 & 0\\ \hline
    Brain (Fig.~\ref{fig:result_lsdem}(e)) & 25000 & 16.5567 & 1.4136 & 0.6063  &  0.0876& 3.9402  & 0.0188 & 0 \\ \hline
  \end{tabular}}
\end{center}
\end{table}

To analyze the performance of our LSDEM algorithm more quantitatively, Table~\ref{tab:LSDEM} shows the computational time, the variance of the initial density and final density, the norm of the Beltrami coefficient, the initial and final landmark mismatch error, and the number of overlaps for mapping different surfaces. Note that because of the additional landmark constraints in the LSDEM problem formulation, it is expected that the density-equalizing effect will not be as good as the one achieved by SDEM. Nevertheless, we can see in all experimental results that the variance of the density is significantly reduced under the LSDEM algorithm, which indicates that the mappings are close to density-equalizing. We also see that the landmark mismatch error is very small and the quasi-conformal distortion is low in all examples. Moreover, the number of overlaps is again 0 in all experiments, which indicates that our LSDEM results are all bijective.

It is natural to ask how the choices of the parameters $\alpha, \beta, \gamma$ in the energy~\eqref{eqt:combined} will affect the mapping results in the proposed LSDEM algorithm. Here, we consider mapping the David model with different values of $\alpha, \beta, \gamma$ used and analyzing the mapping results in terms of the density-equalizing effect, quasi-conformal distortion, landmark mismatch error, and bijectivity. As shown in Table~\ref{tab:LSDEM_para}, if we increase the value of $\alpha$ while keeping $\beta$ and $\gamma$ fixed, we will be able to reduce the variance of the final density. In other words, we will achieve a better density-equalizing effect. By contrast, if we increase the value of $\beta$ while keeping $\alpha$ and $\gamma$ fixed, the quasi-conformal distortion will be reduced. If we increase the value of $\gamma$ while keeping $\alpha$ and $\beta$ unchanged, we will achieve a lower landmark mismatch error. It is noteworthy that the mapping results contain no overlaps for all combinations of $\alpha, \beta, \gamma$, which indicates that the bijectivity of the mappings is ensured under our LSDEM method. 

\begin{table}[t!]
\small
    \caption{\textbf{The performance of our LSDEM algorithm with different parameters.} Here, $\alpha, \beta, \gamma$ are the parameters in the energy~\eqref{eqt:combined}. See the caption of Table~\ref{tab:LSDEM} for the description of the other columns. }\label{tab:LSDEM_para}
  \begin{center}
  \setlength{\tabcolsep}{4mm}{
  \begin{tabular}{|c|c|c|c|c|c|c|c|} \hline 
    \bf $\alpha$ & \bf $\beta$ & \bf $\gamma$
      & \bf $\text{Var}(\widetilde{\rho}_2)$   & \bf Mean$(|\mu|)$  & \bf $\textbf{LE}_2$  & \bf \# Overlaps  \\\hline 
    1 & \multirow{3}{*}{1}  & \multirow{3}{*}{1} & 0.0027 & 0.1687 & 0.2412 & 0  \\ \cline{1-1} \cline{4-7}
    3 &  &  & 0.0009 & 0.1720 & 0.3006 & 0  \\ \cline{1-1} \cline{4-7}
    5 &  &  & 0.0004 & 0.1697 & 0.2292 & 0 \\ 
    \hline
    \multirow{3}{*}{1} & 1 & \multirow{3}{*}{1} & 0.0027 & 0.1687 & 0.2412 & 0  \\ \cline{2-2} \cline{4-7}
     & 3 &  & 0.0159 & 0.1468 & 0.1670 & 0 \\ \cline{2-2} \cline{4-7}
     & 5 &  & 0.0573 & 0.1195 & 0.2576 & 0 \\ 
    \hline
    \multirow{3}{*}{1} & \multirow{3}{*}{1} & 1 & 0.0027 & 0.1687 & 0.2412 & 0  \\ \cline{3-3} \cline{4-7}
     &  & 3  & 0.0024 & 0.1690 & 0.1939 & 0 \\ \cline{3-3} \cline{4-7}
     &  & 5  & 0.0064 & 0.1527 & 0.0714 & 0 \\ \hline
  \end{tabular}}
\end{center}
\end{table}

\section{Applications}\label{sec:applications}

In this section, we introduce the applications of our proposed SDEM and LSDEM methods for surface registration, surface remeshing, and spherical data visualization.

\begin{figure}[t]
    \centering
    \includegraphics[width=\textwidth]{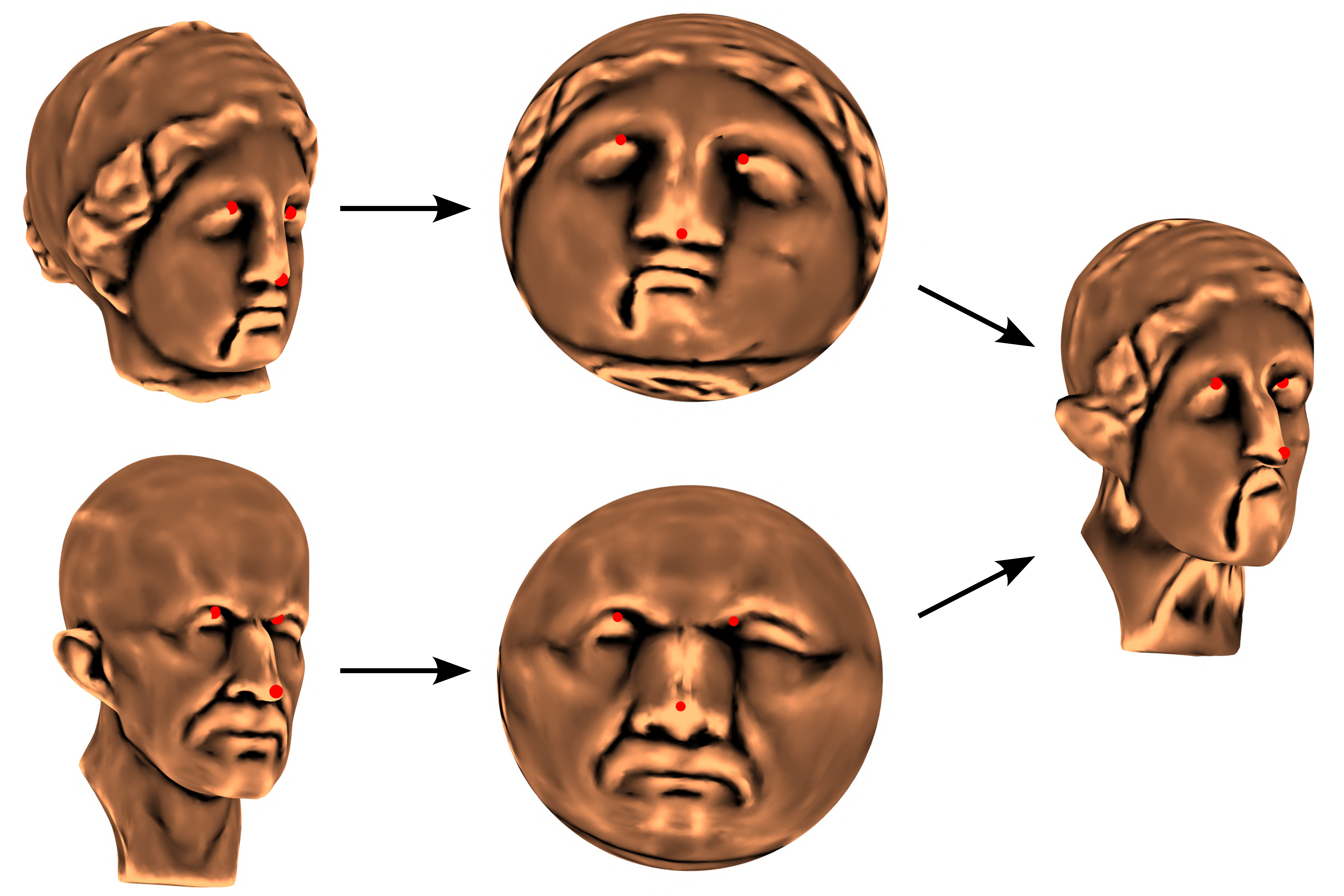}
    \caption{\textbf{Surface registration via our proposed LSDEM method.} We first label some consistent features on both the Igea model and the Max Planck model as landmarks. Then, we compute a landmark-aligned spherical density-equalizing map for each model, with the two sets of landmarks mapped to consistent locations on the unit sphere. We can then use the spherical parameterization to register the Igea model with the Max Planck surface, with the eyes and noses well aligned.}
    \label{fig:registration}
\end{figure}

\subsection{Surface registration}
Suppose $\mathcal{M}$, $\mathcal{N}$ are two genus-0 closed surfaces with some corresponding features $\{p_i\}_{i=1}^k \leftrightarrow \{q_i\}_{i=1}^k$. Using the proposed LSDEM method, we can find two spherical parameterizations $f: \mathcal{M} \to \mathbb{S}^2$ and $g: \mathcal{N} \to \mathbb{S}^2$ with the landmarks well aligned (i.e., $f(p_i) \approx g(q_i)$ for all $i = 1, 2, \dots, k$). Then, we can use the inverse mapping to build a 1-1 correspondence between $\mathcal{M}$ and $\mathcal{N}$. More explicitly, the composition $g^{-1} \circ f: \mathcal{M} \to \mathcal{N}$ will be a landmark-aligned mapping between the two surfaces. Fig.~\ref{fig:registration} shows an example of registering the Igea model and the Max Planck model, in which we use the eyes and noses as landmarks for computing the landmark-aligned spherical parameterizations. We can then obtain the registration mapping from Igea onto Max Planck, with the eyes and noses well matched in the registration result.

\begin{figure}[t!]
    \centering
    \includegraphics[width=\textwidth]{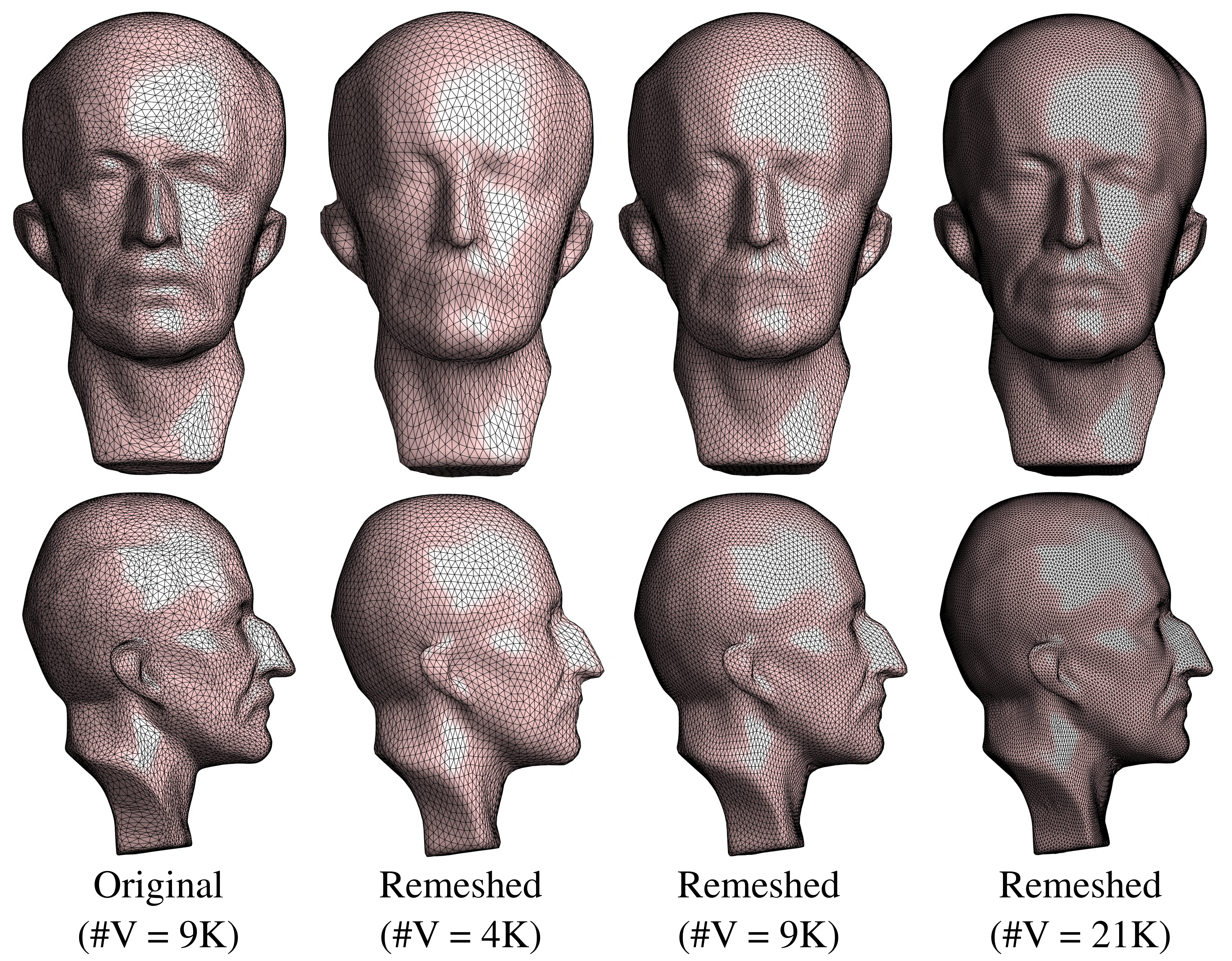}
    \caption{\textbf{Surface remeshing via our proposed SDEM method.} The first column shows the front view and side view of the original Max Planck Model. The second, third, and fourth columns show the front view and side view of three remeshed surfaces with different number of vertices obtained via SDEM.}
    \label{fig:remeshing}
\end{figure}

\subsection{Surface remeshing}
Using the proposed SDEM method, we can perform surface remeshing for genus-0 closed surfaces easily. More specifically, given a genus-0 closed surface $\mathcal{M}$ to be remeshed, we can first compute a spherical density-equalizing map $f: \mathcal{M} \to \mathbb{S}^2$. Then, we can generate a uniform mesh on the sphere using DistMesh~\cite{persson2004simple}. Using the inverse mapping $f^{-1}$, we can then map the uniform mesh structure on the sphere back onto the given surface $\mathcal{M}$, which gives a remeshed surface.

\begin{figure}[t]
    \centering
    \includegraphics[width=\textwidth]{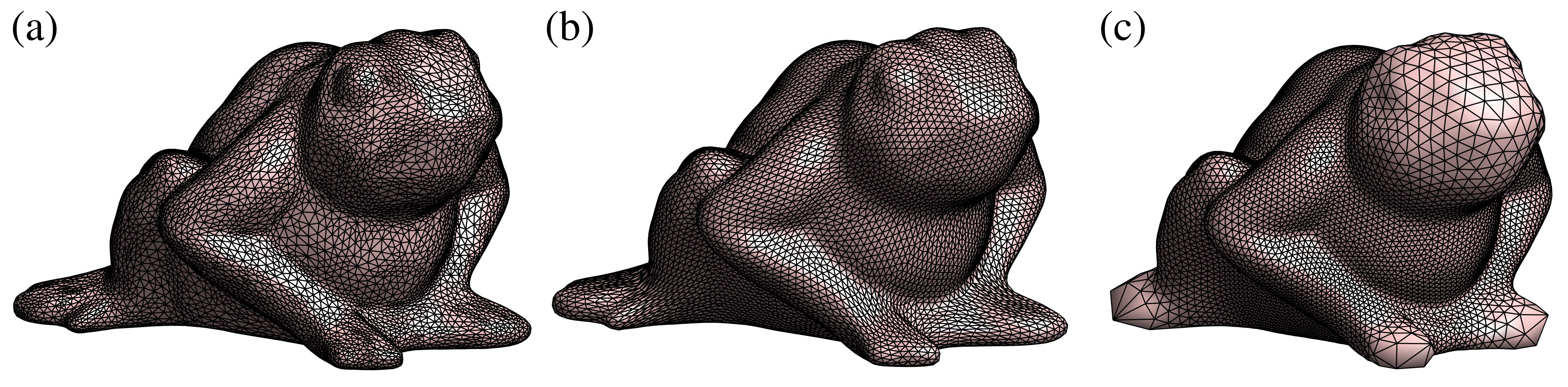}
    \caption{\textbf{A comparison between our SDEM method and conformal parameterization for surface remeshing.} (a) The original Frog model. (b) A remeshed surface obtained via spherical area-preserving parameterization by SDEM. (c) A remeshed surface obtained via spherical conformal parameterization~\cite{choi2015flash,choi2020parallelizable}. The two remeshed surfaces contain the same number of vertices and faces.}
    \label{fig:remeshing_comparison}
\end{figure}

In particular, by computing an area-preserving parameterization using SDEM, we can ensure that the mesh density will be largely uniform in the surface remeshing result. Fig.~\ref{fig:remeshing} shows several examples of remeshing the Max Planck model with different number of vertices. It can be observed that the mesh quality of the remeshed surfaces is very high, with the triangle elements being uniform in size and distribution. In Fig.~\ref{fig:remeshing_comparison}, we further compare our remeshing result with the remeshing result obtained using conformal parameterization~\cite{choi2015flash,choi2020parallelizable}. It can be observed that the triangle elements on the remeshed surface obtained via our method are much more uniform than those obtained via conformal parameterization, which can be explained by the fact that the conformal parameterizations preserve angles but may yield a large distortion in area. Moreover, the sharp features such as the head and the feet of the Frog model are much more well-preserved under our approach. This shows that our method is more advantageous for surface remeshing.

\begin{figure}[t]
    \centering
    \includegraphics[width=0.8\textwidth]{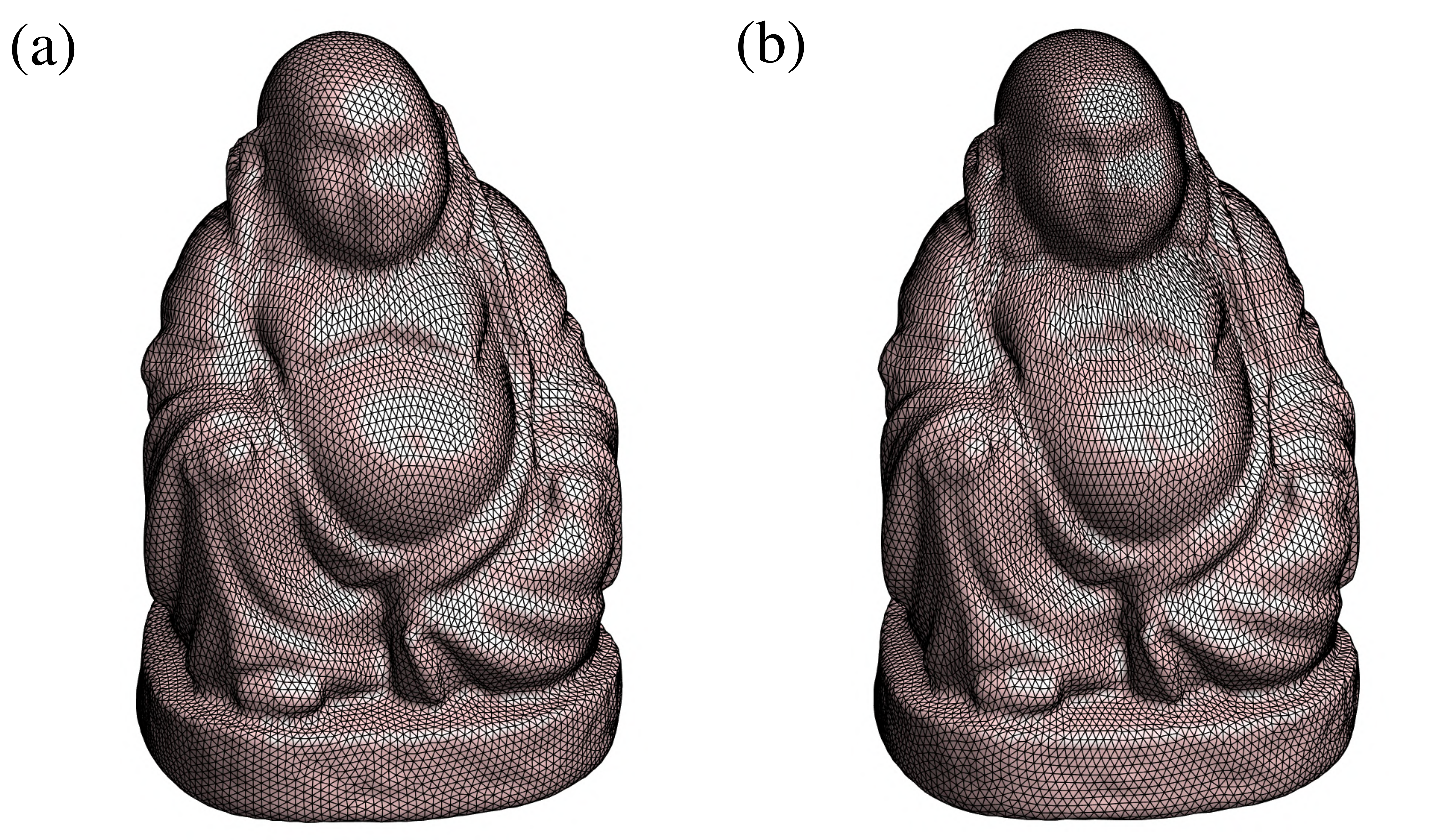}
    \caption{\textbf{Controlling the mesh density in surface remeshing via SDEM.} (a) A remeshed surface of the Buddha model obtained via spherical area-preserving parameterization by SDEM. (b) A remeshed surface of the Buddha model obtained via spherical density-equalizing map with the head part enlarged by SDEM. The two remeshed surfaces contain the same number of vertices and faces.}
    \label{fig:remeshing_enlarge}
\end{figure}

One can also control the mesh density at a certain region by setting an appropriate population in the computation of SDEM. Specifically, if a higher mesh density is desired at a certain region, we can set a larger population and run the SDEM algorithm. Then, the region will be enlarged in the resulting spherical parameterization. Consequently, if we generate a uniform mesh on the sphere using DistMesh and map the new mesh structure back onto the given surface, the mesh density at the corresponding region on the given surface will be higher than the other regions. As an example, Fig.~\ref{fig:remeshing_enlarge}(a) shows a remeshed surface of the Buddha model obtained by computing a spherical area-preserving parameterization using SDEM, in which we can see that the triangle elements are uniform in size and distribution. By contrast, Fig.~\ref{fig:remeshing_enlarge}(b) shows another remeshed surface obtained by computing a spherical density-equalizing map with a larger population set at the head region in the SDEM computation. It can be observed that since the head region is enlarged in the resulting spherical parameterization, the mesh density at the head region of the new mesh structure induced by the inverse mapping will be naturally higher than that at other parts of the remeshed surface.

\subsection{Spherical data visualization}
Analogous to the traditional DEM method~\cite{gastner2004diffusion}, the SDEM method can be applied for visualizing sociological data on a spherical domain. Specifically, we can compute spherical density-equalizing maps on a spherical earth model based on certain given sociological data. The deformed earth models then provide us with an intuitive way to visualize, understand, and analyze the data. Moreover, as the mappings produced by the SDEM method are bijective, no regions will be overlapping in the deformed models. 

\begin{figure}[t]
    \centering
    \includegraphics[width=\textwidth]{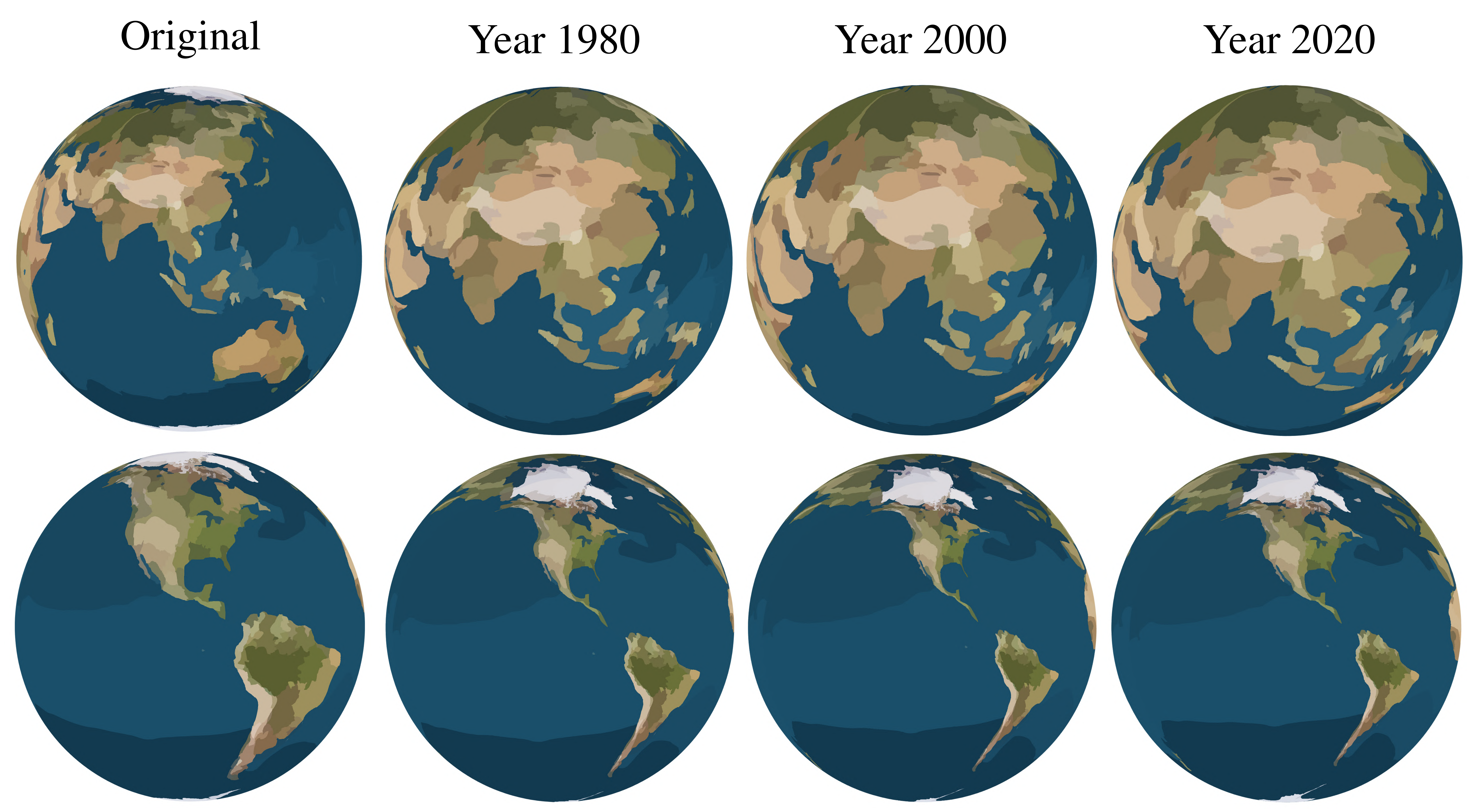}
    \caption{\textbf{Visualizing the change in the world population via SDEM.} The first column shows the undeformed earth model. The second, third, and fourth columns show the deformed earth models produced by SDEM based on the world population in Year 1980, 2000, and 2020, respectively. Two different views are displayed for each earth model (top row and bottom row).}
    \label{fig:visualization_population}
\end{figure}

To illustrate this idea, in Fig.~\ref{fig:visualization_population} we compute three spherical density-equalizing maps based on the world population in Year 1980, 2000, and 2020, respectively. More specifically, for each set of world population data in a specific year, we set the initial densities at different regions (Asia, Africa, North America, South America etc.) of the spherical earth model to be proportional to the actual human population in those regions. Also, we set the density at the sea to be the average of the density at all other regions. Then, under the proposed SDEM algorithm, regions with a higher actual human population will expand and those with a lower actual human population will shrink. As the human population in Asia is higher than that in North America and South America, we can see that the Asia region is magnified in all three mapping results in Fig.~\ref{fig:visualization_population}, while the North America and South America regions are shrunk. Moreover, by comparing the three mapping results for Year 1980, 2000, and 2020, it can be observed that the Asia region in the deformed earth models becomes larger and larger from Year 1980 to 2020. This reflects the change in demographics over the past several decades, with the overall human population in Asia increasing more rapidly than in North and South America.

\begin{figure}[t!]
    \centering
    \includegraphics[width=\textwidth]{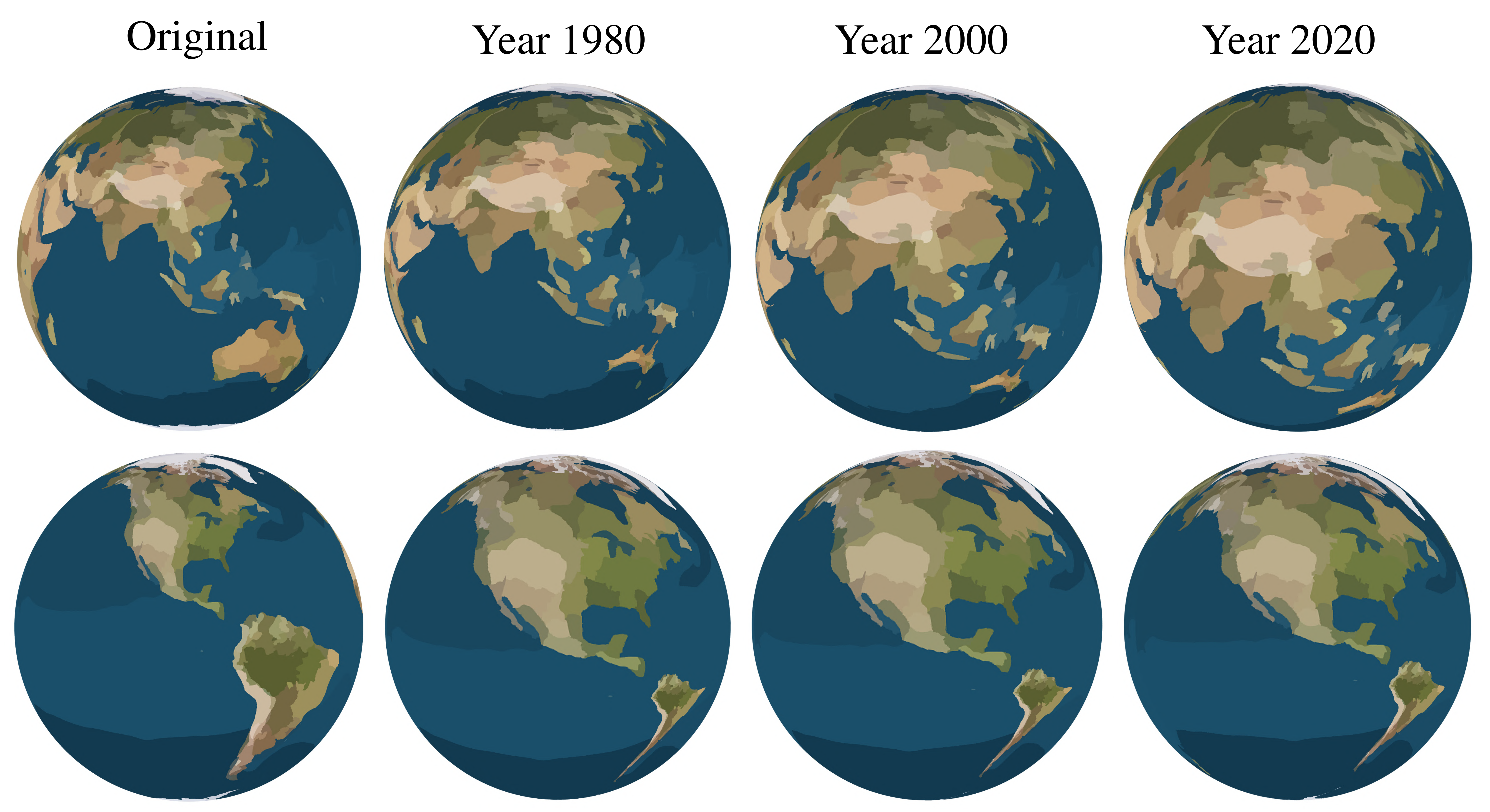}
    \caption{\textbf{Visualizing the change in the nominal gross domestic product (GDP) via SDEM.} The first column shows the undeformed earth model. The second, third, and fourth columns show the deformed earth models produced by SDEM based on the nominal GDP of different regions in Year 1980, 2000, and 2020, respectively. Two different views are displayed for each earth model (top row and bottom row).}
    \label{fig:visualization_gdp}
\end{figure}

In Fig.~\ref{fig:visualization_gdp}, we show another set of spherical density-equalizing maps computed based on the nominal gross domestic product (GDP) of different regions (Asia, Africa, North America, South America etc.) in Year 1980, 2000, and 2020, respectively. This time, we can see that North America is enlarged in all three mapping results, which shows that the nominal GDP of North America has remained relatively high over the past several decades. By contrast, the size of the Asia region shows an increasing trend in the three mapping results from Year 1980 to 2020, which reflects that economic growth in Asia has been emerging over the years. The two sets of examples above show that our method can be effectively applied to sociological data visualization.

 \section{Discussion}\label{sect:discussion}
In this work, we have proposed a novel method for computing spherical density-equalizing maps (SDEM) for genus-0 closed surfaces. Using the proposed method, we can easily achieve different mapping effects, with a particular example being spherical area-preserving parameterizations. We have further proposed a combined energy model including density gradient, harmonic energy, and landmark mismatch energy to achieve landmark-aligned spherical density-equalizing maps (LSDEM) balancing different distortions. The experimental results and applications presented have demonstrated the effectiveness of our proposed methods.

Note that both of the two proposed mapping methods are limited to surfaces with spherical topology. In the future, we plan to extend the methods for surfaces with other topologies. In particular, as many real-world surfaces are high-genus, having both a density-equalizing mapping method and a landmark-aligned density-equalizing mapping method for them would facilitate their processing and shape analysis. Also, note that the proposed mapping methods in the current work are only applicable to triangle mesh data. Another possible future direction is to extend the proposed methods for point cloud surfaces.

\bibliographystyle{ieeetr}
\bibliography{sphericalDEMbib.bib}

\begin{thebibliography}{10}

\bibitem{gu2004genus}
X.~Gu, Y.~Wang, T.~F. Chan, P.~M. Thompson, and S.-T. Yau, ``Genus zero surface
  conformal mapping and its application to brain surface mapping,'' {\em IEEE
  Trans. Med. Imaging}, vol.~23, no.~8, pp.~949--958, 2004.

\bibitem{chen2013ricci}
X.~Chen, H.~He, G.~Zou, X.~Zhang, X.~Gu, and J.~Hua, ``Ricci flow-based
  spherical parameterization and surface registration,'' {\em Comput. Vis.
  Image Underst.}, vol.~117, no.~9, pp.~1107--1118, 2013.

\bibitem{crane2013robust}
K.~Crane, U.~Pinkall, and P.~Schr{\"o}der, ``Robust fairing via conformal
  curvature flow,'' {\em ACM Trans. Graph.}, vol.~32, no.~4, pp.~1--10, 2013.

\bibitem{choi2015flash}
P.~T. Choi, K.~C. Lam, and L.~M. Lui, ``{FLASH}: Fast landmark aligned
  spherical harmonic parameterization for genus-0 closed brain surfaces,'' {\em
  SIAM J. Imaging Sci.}, vol.~8, no.~1, pp.~67--94, 2015.

\bibitem{choi2016spherical}
G.~P.-T. Choi, K.~T. Ho, and L.~M. Lui, ``Spherical conformal parameterization
  of genus-0 point clouds for meshing,'' {\em SIAM J. Imaging Sci.}, vol.~9,
  no.~4, pp.~1582--1618, 2016.

\bibitem{yueh2017efficient}
M.-H. Yueh, W.-W. Lin, C.-T. Wu, and S.-T. Yau, ``An efficient energy
  minimization for conformal parameterizations,'' {\em J. Sci. Comput.},
  vol.~73, no.~1, pp.~203--227, 2017.

\bibitem{choi2020parallelizable}
G.~P.~T. Choi, Y.~Leung-Liu, X.~Gu, and L.~M. Lui, ``Parallelizable global
  conformal parameterization of simply-connected surfaces via partial
  welding,'' {\em SIAM J. Imaging Sci.}, vol.~13, no.~3, pp.~1049--1083, 2020.

\bibitem{liao2022convergence}
W.-H. Liao, T.-M. Huang, W.-W. Lin, and M.-H. Yueh, ``Convergence analysis of
  {D}irichlet energy minimization for spherical conformal parameterizations,''
  {\em J. Sci. Comput.}, vol.~98, no.~29, pp.~1--28, 2024.

\bibitem{angenent2000area}
S.~Angenent, S.~Haker, A.~Tannenbaum, and R.~Kikinis, ``On area preserving
  mappings of minimal distortion,'' in {\em Syst. Theory}, pp.~275--286,
  Springer, 2000.

\bibitem{zou2011authalic}
G.~Zou, J.~Hu, X.~Gu, and J.~Hua, ``Authalic parameterization of general
  surfaces using {L}ie advection,'' {\em IEEE Trans. Vis. Comput. Graph.},
  vol.~17, no.~12, pp.~2005--2014, 2011.

\bibitem{su2013area}
Z.~Su, W.~Zeng, R.~Shi, Y.~Wang, J.~Sun, and X.~Gu, ``Area preserving brain
  mapping,'' in {\em Proceedings of the IEEE Conference on Computer Vision and
  Pattern Recognition}, pp.~2235--2242, 2013.

\bibitem{pumarola20193dpeople}
A.~Pumarola, J.~Sanchez-Riera, G.~P.~T. Choi, A.~Sanfeliu, and
  F.~Moreno-Noguer, ``{3DPeople}: Modeling the geometry of dressed humans,'' in
  {\em Proceedings of the IEEE International Conference on Computer Vision},
  pp.~2242--2251, 2019.

\bibitem{cui2019spherical}
L.~Cui, X.~Qi, C.~Wen, N.~Lei, X.~Li, M.~Zhang, and X.~Gu, ``Spherical optimal
  transportation,'' {\em Comput. Aided Des.}, vol.~115, pp.~181--193, 2019.

\bibitem{yueh2019novel}
M.-H. Yueh, T.~Li, W.-W. Lin, and S.-T. Yau, ``A novel algorithm for
  volume-preserving parameterizations of 3-manifolds,'' {\em SIAM J. Imaging
  Sci.}, vol.~12, no.~2, pp.~1071--1098, 2019.

\bibitem{lui2007landmark}
L.~M. Lui, Y.~Wang, T.~F. Chan, and P.~Thompson, ``Landmark constrained genus
  zero surface conformal mapping and its application to brain mapping
  research,'' {\em Applied Numerical Mathematics}, vol.~57, no.~5-7,
  pp.~847--858, 2007.

\bibitem{nadeem2016spherical}
S.~Nadeem, Z.~Su, W.~Zeng, A.~Kaufman, and X.~Gu, ``Spherical parameterization
  balancing angle and area distortions,'' {\em IEEE Trans. Vis. Comput.
  Graph.}, vol.~23, no.~6, pp.~1663--1676, 2016.

\bibitem{choi2016fast}
G.~P.-T. Choi, M.~H.-Y. Man, and L.~M. Lui, ``Fast spherical quasiconformal
  parameterization of genus-0 closed surfaces with application to adaptive
  remeshing,'' {\em Geom. Imaging Comput.}, vol.~3, no.~1--2, pp.~1--29, 2016.

\bibitem{wang2016bijective}
C.~Wang, X.~Hu, X.~Fu, and L.~Liu, ``Bijective spherical parametrization with
  low distortion,'' {\em Comput. Graph.}, vol.~58, pp.~161--171, 2016.

\bibitem{wang2018novel}
Z.~Wang, Z.~Luo, J.~Zhang, and E.~Saucan, ``A novel local/global approach to
  spherical parameterization,'' {\em J. Comput. Appl. Math.}, vol.~329,
  pp.~294--306, 2018.

\bibitem{lyu2023two}
Z.~Lyu, Q.~Chen, and L.~M. Lui, ``A two-stage algorithm for combined
  quasiconformal and optimal mass transportation spherical parameterization,''
  {\em Math. Comput. Geom. Data}, vol.~3, no.~1, pp.~29--57, 2023.

\bibitem{gastner2004diffusion}
M.~T. Gastner and M.~E.~J. Newman, ``Diffusion-based method for producing
  density-equalizing maps,'' {\em Proc. Natl. Acad. Sci.}, vol.~101, no.~20,
  pp.~7499--7504, 2004.

\bibitem{tobler2004thirty}
W.~Tobler, ``Thirty five years of computer cartograms,'' {\em Ann. Am. Assoc.
  Geogr.}, vol.~94, no.~1, pp.~58--73, 2004.

\bibitem{wake2008we}
D.~B. Wake and V.~T. Vredenburg, ``Are we in the midst of the sixth mass
  extinction? {A} view from the world of amphibians,'' {\em Proc. Natl. Acad.
  Sci.}, vol.~105, no.~Supplement 1, pp.~11466--11473, 2008.

\bibitem{glynn2010breast}
R.~W. Glynn, C.~Scutaru, M.~J. Kerin, and K.~J. Sweeney, ``Breast cancer
  research output, 1945-2008: a bibliometric and density-equalizing analysis,''
  {\em Breast Cancer Res.}, vol.~12, no.~6, p.~R108, 2010.

\bibitem{pratt2012implications}
M.~Pratt, O.~L. Sarmiento, F.~Montes, D.~Ogilvie, B.~H. Marcus, L.~G. Perez,
  R.~C. Brownson, and {the Lancet Physical Activity Series Working Group},
  ``The implications of megatrends in information and communication technology
  and transportation for changes in global physical activity,'' {\em Lancet},
  vol.~380, no.~9838, pp.~282--293, 2012.

\bibitem{matthews2014national}
H.~D. Matthews, T.~L. Graham, S.~Keverian, C.~Lamontagne, D.~Seto, and T.~J.
  Smith, ``National contributions to observed global warming,'' {\em Environ.
  Res. Lett.}, vol.~9, no.~1, p.~014010, 2014.

\bibitem{nusrat2016state}
S.~Nusrat and S.~Kobourov, ``The state of the art in cartograms,'' {\em Comput.
  Graph. Forum}, vol.~35, no.~3, pp.~619--642, 2016.

\bibitem{dodd2016global}
P.~J. Dodd, C.~Sismanidis, and J.~A. Seddon, ``Global burden of drug-resistant
  tuberculosis in children: a mathematical modelling study,'' {\em Lancet
  Infect. Dis.}, vol.~16, no.~10, pp.~1193--1201, 2016.

\bibitem{gastner2018fast}
M.~T. Gastner, V.~Seguy, and P.~More, ``Fast flow-based algorithm for creating
  density-equalizing map projections,'' {\em Proc. Natl. Acad. Sci.}, vol.~115,
  no.~10, pp.~E2156--E2164, 2018.

\bibitem{li2018diffusion}
Z.~Li and S.~Aryana, ``Diffusion-based cartogram on spheres,'' {\em Cartogr.
  Geogr. Inf. Sci.}, vol.~45, no.~5, pp.~464--475, 2018.

\bibitem{li2019visualization}
Z.~Li and S.~A. Aryana, ``Visualization of subsurface data using
  three-dimensional cartograms,'' in {\em Advances in Remote Sensing and Geo
  Informatics Applications: Proceedings of the 1st Springer Conference of the
  Arabian Journal of Geosciences (CAJG-1), Tunisia 2018}, pp.~17--19, Springer,
  2019.

\bibitem{choi2021volumetric}
G.~P.~T. Choi and C.~H. Rycroft, ``Volumetric density-equalizing reference maps
  with applications,'' {\em J. Sci. Comput.}, vol.~86, no.~3, p.~41, 2021.

\bibitem{choi2018density}
G.~P.~T. Choi and C.~H. Rycroft, ``Density-equalizing maps for simply connected
  open surfaces,'' {\em SIAM J. Imaging Sci.}, vol.~11, no.~2, pp.~1134--1178,
  2018.

\bibitem{choi2020area}
G.~P.~T. Choi, B.~Chiu, and C.~H. Rycroft, ``Area-preserving mapping of {3D}
  carotid ultrasound images using density-equalizing reference map,'' {\em IEEE
  Trans. Biomed. Eng.}, vol.~67, no.~9, pp.~1507--1517, 2020.

\bibitem{lyu2023bijective}
Z.~Lyu, G.~P.~T. Choi, and L.~M. Lui, ``Bijective density-equalizing
  quasiconformal map for multiply-connected open surfaces,'' {\em SIAM J.
  Imaging Sci.}, to appear.

\bibitem{lehto1973quasiconformal}
O.~Lehto and K.~I. Virtanen, {\em Quasiconformal mappings in the plane},
  vol.~126.
\newblock Springer Berlin, Heidelberg, 1973.

\bibitem{ahlfors2006lectures}
L.~V. Ahlfors, {\em Lectures on quasiconformal mappings}, vol.~38.
\newblock American Mathematical Society, Providence, RI, 2006.

\bibitem{pinkall1993computing}
U.~Pinkall and K.~Polthier, ``Computing discrete minimal surfaces and their
  conjugates,'' {\em Exp. Math.}, vol.~2, no.~1, pp.~15--36, 1993.

\bibitem{lui2013texture}
L.~M. Lui, K.~C. Lam, T.~W. Wong, and X.~Gu, ``Texture map and video
  compression using {B}eltrami representation,'' {\em SIAM J. Imaging Sci.},
  vol.~6, no.~4, pp.~1880--1902, 2013.

\bibitem{common}
A.~Jacobson, ``Common {3D} test models.''
  \url{https://github.com/alecjacobson/common-3d-test-models}, 2023.
\newblock Accessed: 2023-10-03.

\bibitem{persson2004simple}
P.-O. Persson and G.~Strang, ``A simple mesh generator in {MATLAB},'' {\em SIAM
  Review}, vol.~46, no.~2, pp.~329--345, 2004.

\end{thebibliography}
\end{document}